\documentclass[12pt]{article}%
\usepackage[english]{babel}
\usepackage[titletoc,title]{appendix}
\usepackage{amssymb}
\usepackage{amsfonts,booktabs}
\usepackage{amsmath}
\usepackage{natbib}
\usepackage{appendix}
\usepackage{pstricks}
\usepackage{pst-plot}
\usepackage{bbm} 
\usepackage[nohead]{geometry}
\usepackage[singlespacing]{setspace}
\usepackage[bottom]{footmisc}
\usepackage{indentfirst}
\usepackage{endnotes}
\usepackage{graphicx}
\usepackage{rotating}
\usepackage{subfigure}
\usepackage{verbatim}
\usepackage[hidelinks]{hyperref}
\usepackage{cleveref}
\usepackage{enumitem}
\usepackage{url}
\usepackage[labelsep=period]{caption}
\usepackage[normalem]{ulem}

\usepackage{titlesec}
\titlespacing*{\paragraph}{0pt}{1.2ex plus .2ex minus .2ex}{1em}

\usepackage{newtxtext,newtxmath}

\usepackage{mathrsfs}
\usepackage{tikz}
\usetikzlibrary{arrows.meta,calc}

\definecolor{uncblue}{RGB}{75, 165, 211}
\definecolor{nberblue}{RGB}{0, 90, 155}
\definecolor{resgreen}{RGB}{34, 84, 67}
\definecolor{pennblue}{RGB}{0, 44, 119}
\definecolor{pennred}{RGB}{152, 30, 50}

\newtheorem{theorem}{Theorem}

\newtheorem{corollary}{Corollary}

\newtheorem{definition}{Definition}
\newtheorem{example}[theorem]{Example}

\newtheorem{lemma}{Lemma}

\newtheorem{proposition}{Proposition}

\newcommand*{\lo}{\underline}
\newcommand*{\hi}{\overline}
\newcommand*{\lt}{\left}
\newcommand*{\rt}{\right}
\newcommand*{\ls}{\leqslant}
\newcommand*{\gs}{\geqslant}

\newcommand*{\lmt}{\rightarrow}
\newcommand*{\tight}{\setlength{\itemsep}{0pt}}

\newcommand*{\nn}{\nonumber}

\newcommand*{\dd}{\mathrm{d}}

\newcommand{\FL}{\color{purple} }

\usepackage{accents}

\newenvironment{proof}[1][Proof]{\noindent\textbf{#1.} }{\ \rule{0.5em}{0.5em}}
\makeatletter
\def\@biblabel#1{\hspace*{-\labelsep}}
\makeatother
\title{Algorithmic Attention and Content Creation on Social Media Platforms\footnote{Chen: SC Johnson Graduate School of Management, Cornell University (\href{mailto:yc2535@cornell.edu}{yc2535@cornell.edu}). Li: Department of Economics, University of North Carolina (\href{mailto:lifei@email.unc.edu}{lifei@email.unc.edu}). Preuss: SC Johnson Graduate School of Management, Cornell University (\href{mailto:preuss@cornell.edu}{preuss@cornell.edu}).  The order of authors is alphabetical. We are grateful to Nageeb Ali, Simon Anderson, Ian Ball, Arjada Bardhi, Heski Bar-Isaac, Gary Biglaiser, \"{O}zlem Bedre-Defolie, Gabriel Carroll, Yeon-Koo Che, Daniel Chen, Dongkyu Chang, Roberto Corrao, Alexandre de Corniere, Tom Cunningham, Jacques Cremer, Rahul Deb, Piotr Dworczak, Mira Frick, Yingni Guo, Andrei Hagiu, Kevin He, Wei He, Shota Ichihashi, Justin Johnson, Bruno Jullien, Tony Ke, Kyungmin Kim, Chiu-Yu Ko, Jiangtao Li, Yingkai Li, Yunan Li, Jeanine Miklos-Thal, Ellen Muir, Robin Ng, Alessandro Pavan, Martin Peitz, Patrick Rey, Ron Siegel, Yangbo Song, Bruno Strulovici, Tat-How Teh, Udayan Vaidya, Jiamingmei Wang, Jiaming Wei, Julian Wright, Kai Hao Yang, Leeat Yariv, Huseyin Yildirim, Nathan Yoder, Hanzhe Zhang, Jidong Zhou and participants at numerous seminars and conferences for their comments. All remaining errors are our own.}}
\author{Yi Chen \quad\quad\quad Fei Li\quad\quad\quad Marcel Preuss}
\date{\today\vspace{10pt}\\ 
}

\begin{document}
\hypersetup{hypertexnames=false}

\maketitle
\thispagestyle{empty}

\begin{abstract}
We study how an advertising-funded social media platform organizes content production and consumption through algorithmic attention allocation. Exposure motivates content production, which attracts further attention and generates endogenous network effects. The optimal algorithm uses discriminatory attention subsidies, recommending content with negative direct consumption surplus to stimulate production even under complete information. Creator exposure exhibits an attention cliff and expands from within-category to cross-category audiences as ability rises. We compare this ad-funded market provision of content with the viewer-optimal allocation and show how viewers' private information limits the platform's ability to convert viewer surplus into advertising revenue, distorting content production and consumption. We also examine how monetary rewards interact with exposure and how advertising returns shape the choice of incentive instruments.

\bigskip


\end{abstract}

\onehalfspacing

\clearpage

\pagenumbering{arabic} 

\setlength{\parskip}{6pt}

\setcounter{page}{1}
\section{Introduction}

\noindent Social media platforms hold a prominent position in the digital economy. More than five billion users worldwide spend over two hours per day consuming user-generated content, including short-form videos, live streams, commentary, tutorials, and entertainment. This massive content consumption is fueled by an unmatched supply of content: on TikTok alone, approximately 34 million videos are uploaded each day. Much of this content is produced by independent creators who seek exposure for the recognition and commercial opportunities it brings.\footnote{Social media makes it easier for independent creators to build audiences and monetize their exposure directly through brand sponsorships and affiliate sales. According to EMARKETER (2026), sponsored content is projected to account for 59\% of U.S. social media creator revenues in 2026, compared with 24.4\% from platform payouts and 8.2\% from affiliate marketing. See \url{https://www.emarketer.com/content/faq-on-creator-economy--how-marketers-stand-2026-}} Social media platforms, in turn, monetize the viewer attention attracted by this content through advertising, leading to vast amounts of ads embedded in users' organic feeds. In 2025, social media advertising spending exceeded USD 200 billion.

To manage user interactions at such scale, major platforms such as TikTok, Instagram, YouTube, or X rely on centralized recommendation algorithms to determine which content and sponsored messages appear in each user's feed.\footnote{Instagram's home feed combines posts from followed accounts with suggested content; its Following and Favorites options introduced in 2022 provide chronological alternatives \citep{Instagram2022}. On TikTok, recommended content is even more prominent.} By directing exposure across heterogeneous content and advertisements, these systems govern how attention is generated, allocated, and monetized.

Designing such systems raises important economic questions about user attention allocation. Attention serves two distinct roles: it generates advertising revenue and shapes creators' exposure and incentives to produce content. At the same time, viewer attention is a scarce resource for the platform, as attracting it requires engaging content that is itself costly to generate, as demonstrated in Figure \ref{fig:attention_feedback}. The platform's recommendation system must therefore balance monetization incentives against the need to provide creators with sufficient exposure to sustain the production of engaging content. 

\begin{figure}[h]
\centering

\begin{tikzpicture}[
    font=\footnotesize,
    box/.style={
        draw=black,
        line width=0.8pt,
        rounded corners=2pt,
        minimum width=3.4cm,
        minimum height=0.95cm,
        align=center,
        inner sep=5pt
    },
    flow/.style={
        -{Latex[length=2.2mm,width=1.5mm]},
        line width=0.85pt,
        rounded corners=2pt
    },
    lab/.style={
        font=\scriptsize,
        align=center,
        inner sep=2pt
    }
]

\node[box] (production) at (-6.1,0)
    {Content production};

\node[box] (attention) at (0,0.8)
    {Viewer attention};

\node[box] (advertising) at (6.1,0)
    {Advertising revenue};

\draw[flow]
    (attention.west)
    -- node[lab,above=4pt]
       {Creator exposure}
    (production.north |- attention.west)
    -- (production.north);

\draw[flow]
    (production.east)
    -- node[lab,below=4pt]
       {Attracts attention}
    (attention.south |- production.east)
    -- (attention.south);

\draw[flow]
    (attention.east)
    -- node[lab,above=4pt]
       {Ad impressions}
    (advertising.north |- attention.east)
    -- (advertising.north);

\end{tikzpicture}

\caption{Attention Generation and Dual Uses}
\label{fig:attention_feedback}

\end{figure}

This paper develops a mechanism-design framework to study how an advertising-funded social media platform uses attention allocation to organize content production and consumption. The platform operates as a two-sided marketplace that hosts content across multiple horizontally differentiated categories and serves viewers and creators. Viewers differ in their interests across categories and in their privately known valuation of content quality, and they incur opportunity cost of attention when consuming both content and sponsored advertisements. Creators differ in category and privately known ability, benefit from exposure, and choose costly effort to determine quality. The platform commits to a \emph{recommendation mechanism} that allocates attention across content and advertisements subject to participation and incentive constraints. The platform faces a \emph{trade-off} between harvesting attention through advertising and cultivating it by strengthening production incentives. The returns to both
uses depend on how the recommendation mechanism shapes
viewer engagement and creator production incentives. 

The platform's attention-allocation problem features
\emph{network externalities}.
Attention allocated to creator exposure has a multiplier effect:
greater exposure induces higher effort and improves content quality,
attracting additional attention from interested viewers.
This additional attention can be monetized through advertising
or used to incentivize other creators, generating a feedback loop. The key to designing the attention allocation is to manage this  attention-production loop by determining who watches whose content and how much advertisement, which creators receive subsidized exposure, how much, and from whom. To leverage network externalities, the platform may consequently direct viewers to content whose realized value falls below their attention cost. We call these distorted recommendations \emph{attention subsidies}.

The optimal attention allocation follows a principle of \emph{network-effect-adjusted (virtual) surplus} maximization. The surplus index captures both the direct monetization value of viewer attention and the indirect multiplier effects operating through the attention-production feedback loop. The optimal design uses discriminatory attention subsidies to shape the strength and scope of these network effects. The complementarity between viewer attention and creator effort generates an \emph{attention cliff}: low-ability creators receive no attention, while those just above the inclusion threshold receive a strictly positive amount. The sources of subsidized exposure also vary with creator ability. Medium-ability creators receive subsidies within their own category, whereas sufficiently high-ability creators receive additional cross-category attention when the marginal return to effort is large. Moreover, the interaction of network externalities and private information generates local \emph{clustering of content quality} among creators with different abilities. The attention cliff compresses variation in attention among creators who attract viewers primarily within their own categories. Under private information, incentive compatibility translates these similar attention rewards into similar content quality, even though more able creators exert less effort to achieve it.

Despite the richness of the allocation, the mechanism admits a simple implementation. On the creator side, exposure depends on realized quality, sustaining effort incentives under private information. On the viewer side, the nested structure allows the platform to present a common ordered feed of content and advertisements. Viewers consume sequentially and choose a stopping point that replicates the optimal allocation. This implementation is natural under the “infinite scroll” interface common on social media platforms, and it yields a simple empirical prediction: lower-ranked content should receive less attention and have lower match quality for a given viewer.

Private information generates efficiency losses that are amplified by the two-sided network effects. To limit information rents, the platform allocates attention more selectively toward higher-value viewers and higher-ability creators than under complete information. As a result, participating creators exert less effort and overall viewer engagement declines. Because attention and effort are complementary across the two sides, these distortions propagate through the network and extend even to the top of the user-type space on both sides. Private information thus limits the platform's ability to leverage network effects, shaping which creators and content categories it can sustain. The impact is particularly pronounced for niche categories. With their limited audience base, even under complete information, only sufficiently high-ability creators can make their inclusion profitable for the platform. Private information raises this threshold further, narrowing the range of viable creators and potentially driving entire content categories out of the platform. 

We also examine how the need to monetize viewer surplus through advertising shapes content provision. Under complete information, the platform can maximize viewer surplus and fully extract it through advertising, aligning content production and consumption with the viewer-optimal benchmark. Incomplete information makes full-surplus extraction impossible and breaks this alignment. Holding the information environment fixed, we compare the ad-funded platform with an advertising-free viewer-welfare maximizing planner's allocation. The platform values the network benefits of creator exposure through the surplus it can monetize; the planner values the full viewer surplus. This monetization wedge can distort production and consumption. Both mechanisms retain a similar segmentation of creators, and additional conditions permit comparisons of their cutoffs, aggregate content quality, and the survival of niche categories.

Finally, we return to the advertising-funded platform and allow direct monetary payments to creators. In general, the optimal mechanism may use both exposure and payment for incentive provision, and the relative reliance on the two instruments can vary across creator types. 
Under a mild condition, the pattern becomes especially simple. For each category, there is a cutoff ability: below the cutoff, creators are incentivized only through exposure; above the cutoff, exposure is held fixed and additional incentives are provided through positive monetary transfers. Thus, monetization programs are targeted toward high-ability creators precisely when additional exposure is a relatively inefficient way to reward them. This pattern resembles the coexistence of algorithmic amplification and direct monetization programs on major platforms such as TikTok and YouTube, which combine algorithmic exposure with direct monetization programs to reward creators.

As advertising revenue per impression increases, viewer attention becomes more valuable for direct monetization. The platform therefore substitutes monetary payments for cross-category exposure when incentivizing creators, reducing recommendations that would otherwise be undesirable to viewers. This result links the profitability of advertising to the platform's choice of incentive instruments and the relevance of its recommendations. It also suggests that improvements in ad monetization, such as more effective advertising targeting, may encourage greater use of creator payment programs and more accurate recommendations.

\paragraph{Contribution to Literature}

This paper belongs to the literature on advertising-supported media and multi-sided platform design. Beginning with \citet{AndersonCoate2005}, the media economics literature examines how advertising revenues shape the market provision of content and its welfare implications; see \citet{AndersonJullien2015} for a survey. Related work studies design and screening in advertising-funded markets \citep{ChoiJeon2023,IchihashiJeonKim2026}. More broadly, platforms organize interactions among market participants through pricing and allocation rules \citep{JullienPavanRysman2021}. We study these questions on social media platforms, where content production is decentralized among independent creators. Compared with traditional media, viewer exposure has become a more important source of monetary and nonpecuniary rewards for independent creators on social media platforms. This gives platforms a powerful instrument for shaping creators' production decisions through algorithmic exposure. Viewer attention thus serves both as a source of advertising revenue and as the fuel for content production, creating the central trade-off in its allocation that underlies our analysis.

We contribute to the growing literature on the economics of social media platforms \citep{AridorJimenezDuranLevySong2024}. Several recent papers take a design perspective, examining how platform policies shape incentives primarily on one side of the market. On the viewer side, \citet{AcemogluOzdaglarSiderius2024} study how algorithms influence content sharing, while \citet{hebert2026engagement} examine how information provision shapes viewer engagement. On the creator side, \cite{deCorniereSarvary2023} examine how content bundling affects publishers' incentives to invest in quality. \citet{QianJain2024} study  quality-contingent content recommendation mechanisms as an incentive instrument to mitigate content producers' moral hazard problem, while \citet{BarIsaacDebMitchell2025} examine how visibility and certification can be used to screen content providers.\footnote{Also see related work that examines how content moderation (e.g., \cite{NgTaylor2025}) and platform governance (e.g., \cite{CasnerTeh2025}) shape content creation and viewership.} Our novelty is to characterize how an advertising-funded platform allocates viewer attention to shape the strength and scope of endogenous network effects, linking production incentives and viewer participation through the complementarity between content production and consumption. Our analysis yields distinctive predictions, such as the attention cliff, discriminatory attention subsidies, and local clustering of content quality.

More broadly, our paper adds to the literature on platform design. A central theme in platform economics is how platforms manage network effects across different groups of users, traditionally through pricing and participation decisions \citep{Caillaud2003, rochet2006two}. Beyond pricing, platforms also shape interactions through rankings and recommendations \citep{JohnsonRhodesWildenbeest2023,IchihashiSmolin2025}, as well as matching rules \citep{DamianoLi2007,GomesPavan2016,GomesPavan2024,ValenzuelaStookey2026}. We broaden matching design to account for production-induced externalities across matches. Directing more viewers to a creator induces greater effort and improves the quality of content received by his existing audience. Thus, the value of a viewer's existing connections depends on how the platform allocates attention to others. Under private information, these production externalities also interact with screening: allocating more attention to one creator type can increase information rents and reduce production incentives for higher types. Our principle of network-effect-adjusted virtual surplus maximization captures both the benefits of improved content quality for existing viewers and the screening costs associated with incentivizing different creator types. The optimal allocation therefore determines not only who interacts with whom, but also the returns to production and the costs of screening creators.

Finally, our analysis relates to mechanism design with externalities and the provision of excludable public goods. Access restrictions can help screen consumers' willingness to pay and sustain public-good provision \citep{Hellwig2003,Norman2004}. \citet{corrao2023nonlinear} study optimal pricing when consumer usage generates additional revenue for the platform.  \citet{MeisnerPillath2024} study profit-maximizing mechanisms for non-rival digital content when users' valuations depend on audience size. Unlike these settings, our model features independent content producers whose effort responds to viewer attention, generating both within-side and cross-side network externalities. Their interaction with private information gives rise to distortions even at the top, consistent with earlier findings in screening models with externalities \citep{OrenSmithWilson1982,Lockwood2000}. We identify a setting in which Myersonian methods can be extended to allocation design with endogenous network effects operating within and across the two sides, and derive implications for optimal allocation and welfare.

\paragraph{Organization} The rest of the paper is organized as follows. Section~\ref{sec:model} sets up the model. Section~\ref{sec:com} studies the optimal algorithm in the complete information benchmark. Section~\ref{sec:inc} characterizes the optimal design under incomplete information and discusses its implementation. Section~\ref{sec:welfare} compares ad-funded content provision with the viewer-optimal benchmark. Section~\ref{sec:money} extends the model by allowing for monetary incentives for creators. Section~\ref{sec:conclusion} concludes.

\section{Model}\label{sec:model}

We model a social media platform as a monopolistic two-sided online marketplace where users produce and consume digital content, such as articles, music, and videos. The platform employs a personalized recommendation algorithm to distribute user-generated content and sponsored advertisements (ads) in order to maximize advertising revenue.

\subsection{Environment}

The platform intermediates between two groups of users: content \emph{creators} (he), who incur costs to produce content and benefit from attracting viewer attention, and content \emph{viewers} (she), who spend attention and derive utility from content but not from advertisements. Content is organized into $N>1$ horizontal \emph{categories}, representing distinct topical domains (e.g., news, entertainment, sports). Let $\mathcal{N}\equiv\{1,\ldots,N\}$.

\paragraph{Content Production}
A creator is characterized by a vertical type called \emph{ability} $\theta\in\Theta\equiv[0,\overline{\theta}]$, and a horizontal \emph{category} $i\in\mathcal{N}$ which is the only category in which he can produce content. There is a continuum of creators in category $i$ with mass $\alpha_i>0$. Let $\alpha\equiv\sum_i\alpha_i$ denote the total mass of creators. Within category $i$, abilities follow a CDF $F_i$ with a continuous PDF $f_i$. We assume strictly increasing hazard rate $\frac{f_i(\theta)}{1-F_i(\theta)}$. We refer to a creator of ability $\theta$ in category $i$ as a $(\theta,i)$-creator, sometimes simply $\theta$-creator when no confusion arises.

A creator exerts effort $e\in\mathbb{R}_+$ to produce one unit of content. At ability level $\theta$, effort increases \emph{quality} of this piece of content according to:
$$q=\theta e.$$
Ability and effort are thus complementary: higher ability enables a higher marginal return to effort. The cost of effort is linear and normalized to $e$.

A creator derives utility from the total mass of viewer attention received. Receiving $X\in\mathbb{R}_+$ mass of attention while exerting effort $e$, a creator's net payoff is:
$$u(X)-e,$$
where $u(0)=0$ and $u'(\cdot)>0$. 
The attention-driven utility $u(X)$ captures benefits from reach that are outside the platform’s direct control. These benefits include external monetization opportunities, such as sponsorships, private advertising deals, live-stream selling, and off-platform sales, as well as non-monetary intrinsic motives, such as psychological satisfaction, social influence, and career advancement.\footnote{The intrinsic motive specification is consistent with the empirical findings in \cite{Srinivasan2024}. Using Reddit data, the paper shows that content creation by anonymous creators remains highly responsive to viewer attention even when direct monetization and reputational opportunities are limited. \cite{ToubiaStephen2013} also document image-related utility as being a significant  motivating factor for Twitter users.
} In Sections~\ref{sec:com} and \ref{sec:inc}, we assume any utility creators receive stems directly from attention.\footnote{Our  baseline model assumes that creators value total attention regardless of its source. In practice, under some platform environments and monetization models, the benefit generated by a viewer’s attention may depend on the viewer’s preferences or audience type. Since our goal is not to commit to a specific business model, we abstract from this heterogeneity in the baseline and discuss its implications in Section~\ref{sec:dis}.
} 
Section~\ref{sec:money} allows creators to receive payment from the platform.

\paragraph{Content Consumption}

A viewer is characterized by a horizontal type $j\in\mathcal{N}$, which represents her \emph{preferred} content category, and a vertical type called \emph{value} $v\in V\equiv[\lo{v},\hi{v}]$, which measures how much she values quality in that category. 
There is a continuum of viewers in category $j$ with mass $\beta_j>0$. Let $\beta\equiv\sum_j\beta_j$ denote the total mass of viewers. Within category $j$, values follow a CDF $G_j$ with a continuous PDF $g_j$ and an increasing hazard rate $\frac{g_j(v)}{1-G_j(v)}$. We refer to a viewer of value $v$ in category $j$ as a $(v,j)$-viewer, sometimes simply $v$-viewer when no confusion arises.

Paying attention is costly: each unit of attention spent on content or ads costs a constant $c>0$, regardless of its use. This \emph{attention cost} captures the opportunity cost and potential ``nuisance cost'' of advertising that often appear in the literature \citep{AndersonCoate2005}. A $(v,j)$-viewer who spends one unit of attention on content of quality $q$ in category $i$ receives net payoff:
\[
vq\mathbbm{1}\{i=j\}-c.
\]
The specification captures horizontal preference differentiation through viewers’ preferred categories, and vertical heterogeneity through how much they value quality in that category.\footnote{Allowing heterogeneous attention costs does not change the analysis, since participation depends only on the ratio $v/c$. We therefore normalize $c$ across viewers and attribute heterogeneity to $v$.} The payoff from consuming content in other categories is normalized to zero to keep the model tractable and isolate our core mechanism. Section~\ref{sec:dis} discusses how this assumption can be relaxed. 
\paragraph{Sponsored Advertising}

The platform profits from viewer exposure to \emph{ads} that are displayed alongside content. We assume the platform earns a constant marginal revenue $z>0$ per unit of attention allocated to advertisements. This revenue can be interpreted as the equilibrium winning bid in an auction among advertisers with homogeneous values for viewer impressions. 

When the viewer spends a unit of attention on ads, we also normalize the benefit to zero, so that the net payoff is simply $-c$. Viewers allocate attention across content and advertisements. Preferences are additive across units of attention, so her total utility is the sum of the net payoffs generated by all units she spends. 

\paragraph{Information Structure}

The platform does not observe users' vertical types \(\theta\) and \(v\). This information asymmetry makes incentive provision on both sides of the platform nontrivial. 
The assumption that vertical types remain private is also natural. A viewer's valuation for quality may reflect residual preference heterogeneity and transient shocks, such as mood, attention, or available leisure time. A creator's effective productivity may likewise vary with changing ideas or external circumstances. For example, a creator who posts content from daily life may find it easier to produce exciting content while traveling. These sources of variation preserve an important element of private information even when the platform has access to rich behavioral data. On the other hand, horizontal categories mainly determine matching and generate little additional incentive conflict between the platform and its users. They are also relatively easy to infer: a creator's category is revealed by the content he posts, while a viewer's favorite topic can be learned from her browsing and engagement history. We therefore assume that the platform observes users' horizontal categories \(i\) and \(j\).

The platform observes the quality $q$ of each piece of content but cannot identify the components $\theta$ and $e$. This mirrors the real-world practice: quickly conducting experiments of various scales on randomly selected viewers and inferring content performance from engagement data. 

\subsection{Algorithm Design Problem}

Viewers and creators do not interact directly. Instead, the platform intermediates both sides through a recommendation algorithm that allocates viewers' attention across content of different categories and quality levels, as well as advertisements. Throughout, we use the term ``algorithm'' to refer to the induced allocation rule rather than the computational procedure that implements it.

\paragraph{Examples}

Before moving further, we discuss strategic interactions on both sides under a few stylized recommendation algorithms. 
These examples show that, by controlling who is exposed to whose content, the algorithm endogenously shapes the direction and strength of network effects to stimulate content production, increasing advertising revenue.

First, suppose that the algorithm recommends to each viewer only content from her preferred category, regardless of the quality. To generate advertising revenue, the platform mixes some ads in the recommended feed. While such an algorithm maximizes relevance from the viewer's perspective, it provides no incentive for creators to exert effort, and the equilibrium quality collapses. Anticipating the lack of quality, viewers would choose not to engage. 

Next, suppose the algorithm recommends only preferred-category content to viewers, subject to a minimum quality threshold for exposure. This policy partially restores incentives on the supply side as high-ability creators exert just enough effort to meet the threshold, while low-ability creators opt out. With meaningful qualities on the platform,  
high-valuation viewers choose to consume the content-ads mix. In turn, viewers' engagement encourages more creators to participate. The algorithm can generate \emph{within-category cross-subsidization}: some lower-valuation viewers receive content whose quality is insufficient to justify their attention costs, but their attention incentivizes creators to produce higher-quality content that benefits higher-valuation viewers. This feedback encourages further creator participation and generates network effects within each category.

Finally, suppose the algorithm recommends every preferred-category content item that exceeds a  threshold, and promotes non-preferred-category content to viewers only if its quality exceeds a higher threshold. Creators now expect an even larger audience at high quality levels, and those with high ability will exert more effort. While such non-preferred-category content  lowers the utility of viewers, the heightened effort from creators indirectly benefits those viewers who \emph{are} interested in those categories. The algorithm thus extends \emph{network effects across categories}, allowing attention from one group of viewers to incentivize content production for another.\footnote{The cross-category subsidy mechanism only requires some creators to value some cross-category exposure. Also, if viewers derive positive utility from consuming some content outside their preferred category, cross-category attention becomes a more appealing incentive instrument for the platform. See Section \ref{sec:dis} for details.}


\paragraph{Recommendation Algorithm}

The preceding examples are useful for building intuition, but they represent only special cases of feasible recommendation rules. In practice, the platform implements personalized recommendations that condition on users' characteristics, selectively leveraging network effects. To characterize the optimal algorithm without loss of generality, we invoke the Revelation Principle \citep{Myerson1986} and look for the optimal direct mechanism. 
An \emph{algorithm} is a (direct) mechanism that elicits users' information and distributes content and ads, while respecting users' incentive compatibility and individual rationality. Formally, an algorithm
$$\mathscr{A}\equiv\lt(\{e_i\}_{i\in\mathcal{N}}, \{x_{ij}\}_{i,j\in\mathcal{N}}, \{a_j\}_{j\in\mathcal{N}}\rt)$$
consists of three sets of functions. Effort recommendation $e_i:\Theta\lmt\mathbb{R}_+$ maps reported types of category-$i$ creators into a recommended effort level, attention assignment $x_{ij}:\Theta\times V\lmt[0,1]$ maps category-$i$ creator's reported type and category-$j$ viewer's reported type into a probability that the creator's content is recommended to the viewer, and finally $a_j:V\lmt\mathbb{R}_+$ maps a category-$j$ viewer's reported type into the total mass of ads mixed in their recommendation bundle. The platform commits to the mechanism before users report and creators exert effort. After a creator reports $\theta'$, the assigned attention profile is conditional on delivering observable quality $q=\theta' e_i(\theta')$; any other realized quality receives zero attention. Creators deliver quality before viewers consume the assigned content and ads. For notational simplicity, define:
\begin{equation*}
X_i(\theta)\equiv\sum_j\beta_j\int_Vx_{ij}(\theta,v)\dd G_j(v)
\end{equation*}
as the total amount of attention a $(\theta,i)$-creator receives, determined by $\{x_{ij}\}_{i,j}$. Then \emph{incentive compatibility} (IC) requires, for any $\theta\neq 0,\theta',i$:
\begin{eqnarray}
&&\hspace*{-2em}u(X_i(\theta))-e_i(\theta)\gs u(X_i(\theta'))-\frac{\theta'e_i(\theta')}{\theta},\label{eq:IC-p}
\end{eqnarray}
and for any $v,v',j$:
\begin{eqnarray}
&&\sum_i\alpha_i\int_{\Theta}x_{ij}(\theta,v)\Big(v\theta e_i(\theta)\mathbbm{1}\{i=j\}-c\Big)\dd F_i(\theta)-ca_j(v)\nn\\
&\gs&\sum_i\alpha_i\int_{\Theta}x_{ij}(\theta,v')\Big(v\theta e_i(\theta)\mathbbm{1}\{i=j\}-c\Big)\dd F_i(\theta)-ca_j(v').\label{eq:IC-c}
\end{eqnarray}
A creator with $\theta=0$ can produce only zero quality and cannot mimic a contract requiring positive quality. Intuitively, constraint \eqref{eq:IC-p} says that a $(\theta,i)$-creator, who can produce content with quality $\theta'e_i(\theta')$ and receive attention according to $\{x_{ij}(\theta',v)\}_{v,j}$, finds it optimal to truthfully report $\theta'=\theta$ and follow the recommended effort $e_i(\theta)$. Similarly, constraint \eqref{eq:IC-c} states that a $(v,j)$-viewer  must prefer recommendation $\{x_{ij}(\theta,v)\}_{\theta,i}$ mixed with $a_j(v)$ amount of ads to the recommendation $\{x_{ij}(\theta,v')\}_{\theta,i}$ mixed with $a_j(v')$ amount of ads, intended for any other $v'$. \emph{Individual rationality} (IR) requires that any creator and viewer must receive payoff greater than their outside option, normalized to zero:
\begin{eqnarray}
&&u(X_i(\theta))-e_i(\theta)\gs0,\ \forall\ \theta,i,\label{eq:IR-p}\\
&&\sum_i\alpha_i\int_{\Theta}x_{ij}(\theta,v)\Big(v\theta e_i(\theta)\mathbbm{1}\{i=j\}-c\Big)\dd F_i(\theta)-ca_j(v)\gs0,\ \forall\ v,j.\label{eq:IR-c}
\end{eqnarray}

The platform maximizes the income from the ads market, that is:
\begin{eqnarray}
\max_{\mathscr{A}}\ z\sum_j\beta_j\int_Va_j(v)\dd G_j(v)\quad \mbox{ s.t.\ }\quad \eqref{eq:IC-p}, \eqref{eq:IC-c}, \eqref{eq:IR-p}, \eqref{eq:IR-c}.\nn
\end{eqnarray}

\subsection{Model Discussion}\label{sec:dis}

We discuss some simplifying assumptions that allow us to isolate the core insight. 
First, we assume that any positive-quality content in a viewer’s preferred category generates positive utility (i.e., $\lo{v}>0$), whereas cross-category content and advertisements generate zero payoff regardless of quality.\footnote{The analysis is essentially unchanged if cross-category content provides positive utility below the attention cost $c$, making cross-category exposure a less costly incentive instrument. } This parsimoniously captures heterogeneity in preferences over digital content on the platform, including both organic and sponsored material. Under this formulation, a viewer’s private information is one-dimensional, which keeps the mechanism design problem tractable. Another advantage of this modeling approach is that it isolates the production-incentive role of cross-category exposure. Since cross-category content does not directly raise viewer utility in the baseline model, whenever the platform exposes viewers to such content, the rationale must come from its effect on creator incentives and the resulting quality improvements elsewhere on the platform.\footnote{ More generally, viewers may value content quality differently across categories. Such heterogeneity introduces additional considerations about which cross-category content to recommend and to whom, beyond the total amount of cross-category exposure. Our baseline abstracts from these considerations to isolate the core incentive mechanism. Numerical example \ref{oa-ex1} in the Online Appendix illustrates that our main qualitative insights remain robust.}

Second, we summarize a creator’s exposure benefit by total exposure and abstract from the exact composition of the audience generating that exposure. The key economic force is that viewer attention has incentive value for creators, whether through psychological gratification, career concerns, or monetizable traffic generated outside the platform. What matters for the platform’s design problem is therefore not that every viewer’s attention is equally valuable to a creator, but that attention from a broad set of viewers, including viewers for whom the content is not the best match, can still motivate content production. This is natural since content is an experience good: a viewer typically spends some attention before learning the content’s quality or relevance, and that engagement may already generate visibility, traffic, or other exposure value for the creator. 

More generally, it is natural to allow a category-$i$ creator to value attention from different viewer categories differently:
$
u\left(\sum_{j\in N}\delta_{ij}\beta_j
\int_V x_{ij}(\theta,v)\,\dd G_j(v)\right)$,
where $\delta_{ij}$ measures the value of attention from category-$j$ viewers to a category-$i$ creator, normalized so that $\delta_{ii}=1$.
When $\delta_{ij}=\delta\in(0,1)$ for all $j\neq i$, cross-category attention remains interchangeable across viewer categories. The platform can therefore continue to summarize cross-category exposure by its total mass, as in our baseline model. This preserves the tractability of the analysis and allows us to focus on when and how much attention the platform should use to subsidize content production.
When $\delta_{ij}$ varies across viewer categories, the composition of cross-category exposure also matters. The platform then faces an additional matching problem in deciding which viewers should provide attention subsidies to which creators. While it is clear that the platform should avoid assignments with $\delta_{ij}\ls 0$, a full characterization of the optimal allocation is analytically intractable. Nevertheless, our numerical example \ref{oa-ex2} in the Online Appendix illustrates that the main insights about attention subsidies and endogenous creator inclusion continue to hold. We therefore view our baseline as providing an analytical account of whether and why cross-category attention subsidies arise, and a starting point for studying how these subsidies should be allocated across viewer categories.


Third, we assume that viewers cannot consume only their preferred-category content while avoiding advertisements and other content in the feed. This reflects the idea that both content and ads are, at least partially, experience goods: viewers derive payoffs only after spending attention and therefore cannot perfectly cherry-pick without incurring attention costs. This assumption is common in the literature (see, e.g., \cite{deCorniereSarvary2023}), and it is particularly appropriate for static content such as short text or photos, where by the time viewers determine whether they like the content, the attention cost has already been largely sunk. In this sense, evaluation and consumption are nearly indistinguishable. Longer videos or texts differ somewhat because viewers may attempt to evaluate content after a few seconds. However, such evaluation remains imperfect, as evidenced by the pervasiveness of clickbait. Moreover, some platforms require viewers to watch advertisements for a minimum duration before skipping, such as non-skippable ads on YouTube and in-feed ads on TikTok. As long as viewers cannot screen content without incurring attention costs, the no-cherry-picking assumption is a reasonable approximation.\footnote{Suppose the platform shows a viewer a mass $a>0$ of ads if the viewer cannot cherry pick at all. Now suppose the user can detect and skip ads with half the attention cost, then the platform can just raise the mass of ads to $2a$ to achieve the same outcome.}

Fourth, we assume a constant marginal cost of attention for viewers and abstract from a hard attention-capacity constraint. This is sufficient to make attention scarce (or costly to acquire) for the platform, which is necessary for our central trade-off between advertising and creator incentives. It also makes viewer payoffs quasi-linear, allowing us to apply standard mechanism design techniques to viewer-side screening. To capture an additional source of scarcity arising from viewers' limited attention capacity, one could introduce convex attention costs.\footnote{Although limited viewer attention is not necessary for our trade-off, it is important in
studying competition for viewers' attention (see, e.g.,
\cite{AmbrusCalvanoReisinger2016,PratValletti2022,anderson2023ad,Chen2026}).} Specifically, a $(v,j)$-viewer's total attention cost becomes $C\left(\sum_i\alpha_i\int_{\Theta} x_{ij}(\theta, v)\dd F_i(\theta)+a_j(v)\right)$  
where $C(\cdot)$ is convex, with a hard attention-capacity constraint as a limiting case. Such an extension would eliminate the tractability offered by the quasi-linear structure and preclude a clean analytical characterization. Nevertheless, our numerical example \ref{oa-ex3} in the Online Appendix suggests that the main insights about attention subsidies and endogenous creator inclusion remain robust.

\section{Complete Information}\label{sec:com}
We first examine the competing uses of viewer attention
between advertising and creator exposure under 
complete information.
This benchmark highlights the analytical challenges posed
by endogenous network effects and their economic implications
for optimal attention allocation.


\subsection{Simplifying the Design Problem}
Under complete information, the platform's problem reads:
\begin{eqnarray}
\max_{\mathscr{A}}\ z\sum_j\beta_j\int_Va_j(v)\dd G_j(v)\quad \mbox{ s.t.\ }\quad \eqref{eq:IR-p}, \eqref{eq:IR-c}.\nn
\end{eqnarray}
The algorithm serves as a recommendation mechanism in which obedience constraints ensure that users voluntarily follow the prescribed attention and effort choices.

We begin by simplifying constraints \eqref{eq:IR-p} and \eqref{eq:IR-c}. First, an optimal mechanism does not generically leave a viewer's participation constraint slack; otherwise, it could profitably increase their ads exposure. Second, for any $(\theta,i)$-creator, it is without loss of generality to bind constraint \eqref{eq:IR-p}; otherwise, the algorithm can extract higher $e_i(\theta)$ to further entertain the viewers who appreciate this creator's content and thereby tolerate more ads. These insights are summarized below.

\begin{lemma}[Binding Constraints]\label{lem:binding}\ \\
Under complete information, there exists an optimal algorithm where both creator and viewer participation constraints \eqref{eq:IR-p} and \eqref{eq:IR-c} are binding.
\end{lemma}

Now we substitute out $a_j(v)$ using the viewers' binding IR constraints, omit the constant multiplier $\frac{z}{c}$, and rewrite the problem as:
\begin{align}
\max_{\substack{\{e_i(\cdot)\}_i,\\ \{x_{ij}(\cdot,\cdot)\}_{i,j}}}&
\sum_{i,j}\alpha_i\beta_j\int_{\Theta}\int_Vx_{ij}(\theta,v)\Big(\underbrace{v\theta e_i(\theta)\mathbbm{1}\{i=j\}-c}_{\mbox{\footnotesize Surplus}}\Big)\dd G_j(v)\dd F_i(\theta)\label{eq:obj}\\
\mbox{s.t. }& e_i(\theta)=u(X_i(\theta)),\ \forall\ \theta,i,\label{eq:eff}\tag{IR-C}
\end{align}

The objective in problem \eqref{eq:obj} is the total viewer surplus from watching content only, which is then extracted by the platform using ads. For each combination of $(\theta,i,v,j)$, $v\theta e_i(\theta)\mathbbm{1}\{i=j\}-c$ represents viewer $(v,j)$'s net payoff from watching creator $(\theta,i)$'s content. A positive net payoff leaves room for the platform to show more ads. Constraint \eqref{eq:eff} pins down a creator's effort. Rigorously speaking, this is a relaxed problem of the platform's original optimization. To ensure that the solution to the relaxed problem satisfies the non-negativity constraint $a_j(v)\gs 0$ for advertising in the original problem, we require:
\begin{equation}
\sum_i\alpha_i\int_{\Theta}x_{ij}(\theta,v)\Big(v\theta e_i(\theta)\mathbbm{1}\{i=j\}-c\Big)\dd F_i(\theta)\gs0,\ \forall\ v,j.\label{eq:pos}\tag{NN-C}
\end{equation}
This is necessary because, after the replacement, $a_j(v)$ is no longer an explicit variable in the optimization. In what follows, we follow the standard relaxation-and-verification approach in mechanism design. We first solve the relaxed problem without imposing \eqref{eq:pos} and then derive conditions under which the solution satisfies this constraint and hence solves the original problem. It turns out the relaxed optimal mechanism satisfies \eqref{eq:pos} when the viewer's attention cost is sufficiently small. 
We also demonstrate our main economic message continues to hold when  \eqref{eq:pos} is binding for some types.

To understand the platform's trade-off in attention allocation, we examine the marginal impact of $x_{ij}(\theta,v)$ 
on the platform's profit. Substituting constraint \eqref{eq:eff} into the objective of problem \eqref{eq:obj}, we decompose the marginal effect of $x_{ij}(\theta,v)$ into two terms below.

\begin{definition}[Network-Effect-Adjusted Surplus]\ \\
In the optimal algorithm, the value of $x_{ij}(\theta,v)$ is determined by the sign of
\begin{eqnarray}
\underbrace{v\theta e_i(\theta)\mathbbm{1}\{i=j\}-c}_{\text{\emph{Local Surplus}}}\ +\ \underbrace{u'(X_i(\theta))\theta\beta_i\int_Vx_{ii}(\theta,v')v'\dd G_i(v')}_{\text{\emph{Network Effect}}}.\label{eq:decom}
\end{eqnarray}
In particular, $x_{ij}(\theta,v)=1$ if expression \eqref{eq:decom} is positive, 
$x_{ij}(\theta,v)=0$ if it is negative, and any value in $[0,1]$ is optimal if it is zero.
\end{definition}

The derivative in expression \eqref{eq:decom} captures two effects when recommending type-$(\theta,i)$ creator's content to viewer $(v,j)$. The first term captures the direct pointwise surplus of $(v,j)$-viewer from watching content produced by $(\theta,i)$-creator. This effect is \emph{local}: it depends only on the creator-viewer pair in question and is independent of other allocations. When positive, this surplus relaxes the participation constraint of viewer $(v,j)$ and thus allows for more inserted ads. The second term captures the indirect network multiplier effect.  
By assigning attention $x_{ij}(\theta,v)=1$, the creator $(\theta,i)$ receives $u'(X_i(\theta))$ marginal utility, which in turn allows the platform to extract more effort. The increased effort then boosts the quality of content by a factor of $\theta$ and benefits all viewers in the same category 
who spend attention on it, eventually enabling more ads to be blended in for these viewers. Due to this non-negative network effect, 
attention can improve profit even if the local surplus alone is negative, provided the production-side externality is strong enough. In sum, pointwise \emph{local} surplus maximization is no longer the right principle to follow.

The optimal algorithm is challenging to solve because problem \eqref{eq:obj} is high-dimensional: with $N$ categories, the algorithm consists of $N^2+2N$ functions. Moreover, due to network effects, 
categories are interlinked, and one cannot simply maximize viewer surplus within each category in isolation. 
Fortunately, we can reduce the dimensionality of the platform's problem by exploiting the symmetry of viewer preferences and the separability of the objective function. 

First, consider cross-category attention $x_{ij}$ where $j\ne i$. Because these viewers derive a constant net utility of $-c$ from such non-preferred content, and because digital content is non-rival, the platform's profit and the creator's incentives are independent from the specific composition of the audience. This perfect substitutability makes our framework particularly well suited to studying when and to what extent the platform should use cross-category attention to subsidize content production. We can therefore aggregate the attention from all those ``mismatched'' viewers into a single variable called creator $(\theta,i)$'s \emph{extended reach}, 
the mass of attention from outside category $i$:
\begin{equation*}
R_i(\theta)\equiv\sum_{j\ne i}\beta_j\int_Vx_{ij}(\theta,v)\dd G_j(v) \in [0,\beta-\beta_i],
\end{equation*}
so that total attention received can be recast as 
$X_i(\theta)=\beta_i\int_Vx_{ii}(\theta,v)\dd G_i(v)+R_i(\theta).$

Second, the constant marginal cost of attention implies that the cost of incentivizing one creator is separable from the cost of incentivizing another. Since the platform extracts the full viewer surplus via binding participation constraints, the total ads profit is proportional to the sum of the viewer surplus generated by creators in each category. Thus, the aggregate maximization problem decomposes into $n$ additive sub-problems, one for each creator category. We can thus characterize the optimal schedule $\{e_i(\cdot), x_{ii}(\cdot), R_i(\cdot)\}$ for each category $i$ independently, as stated in \autoref{lem:sep}.

\begin{lemma}[Separability]\label{lem:sep}\ \\
Consider the sub-problem for each category $i\in\mathcal{N}$,
\begin{eqnarray}
\hspace*{-2em}\max_{e_i,x_{ii},R_i}&&\beta_i\int_{\Theta}\int_V x_{ii}(\theta,v)\Big(v\theta e_i(\theta)-c\Big)\dd G_i(v)\dd F_i(\theta)-c\int_{\Theta}R_i(\theta)\dd F_i(\theta)\label{eq:sep}\\
\hspace*{-2em}\mbox{\emph{s.t.}}&&e_i(\theta)=u\lt(\beta_i\int_V x_{ii}(\theta,v)\dd G_i(v)+R_i(\theta)\rt), \ \forall\ \theta,i\nn
\end{eqnarray}
and denote the solution $\{\hat{e}_i,\hat{x}_{ii},\hat{R}_i\}_i$. There exists $\{e_i^*,x_{ij}^*\}_{i, j}$ such that $e_i^*(\theta)=\hat{e}_i(\theta)$, $x_{ii}^*(\theta,v)=\hat{x}_{ii}(\theta,v)$, $\sum_{j\ne i}\beta_j\int_Vx_{ij}^*(\theta,v)\dd G_j(v)=\hat{R}_i(\theta)$, $\forall\ \theta,i,v$, and solves problem \eqref{eq:obj} without \eqref{eq:pos}. 
\end{lemma}

Each subproblem is indexed by the creators' category, and its objective measures the surplus their content generates across all viewers. Within-category recommendations generate viewer surplus and production incentives. Extended reach uses attention from other categories: its cost enters the objective through $-cR_i$, while its contribution to creator effort enters through the constraint. Each subproblem thus jointly optimizes the use of within- and cross-category network effects. 

The reduction to $\{e_i,x_{ii},R_i\}$ makes it possible to characterize which creators participate, how much exposure they receive, and when the platform expands their audience beyond their own category. The construction proceeds in two steps. First, we jointly optimize $x_{ii}$ and $R_i$ for each creator category, with effort determined by the constraint. Second, we assign cross-category recommendations $x_{ij}$, $j\ne i$, across viewer categories and types to deliver the resulting $R_i$. The second step can be constrained by  \eqref{eq:pos}. Each viewer must obtain enough surplus from own-category content to cover the attention cost of the cross-category recommendations assigned to her. Binding constraints can therefore link the allocation of attention subsidies across categories.

\subsection{Optimal Design under Complete Information}

We now characterize the optimal allocation when recommendations condition on users' categories and vertical types. A key insight from problem \eqref{eq:sep} is that within-category attention $x_{ii}$ and creator effort $e_i$ are complementary. If a creator reaches only a subset of viewers in his category, the participation constraint limits his effort and thus content quality. Expanding within-category attention relaxes this constraint, inducing higher effort and improving quality. Importantly, the resulting quality gains benefit not only newly reached viewers but also those already paying attention. This feedback between attention and effort generates increasing returns to within-category attention.

\begin{proposition}[Algorithm: Complete Information]\label{prop:complete}\ \\
Suppose the attention cost $c$ is small. For each category $i$, the optimal algorithm is characterized by a nondecreasing total attention function $X_i(\cdot)$ such that:
$$x_{ii}(\theta,v)=\mathbbm{1}\lt\{X_i(\theta)>0,\ G_i(v)\gs 1-\tfrac{X_i(\theta)}{\beta_i}\rt\},\quad R_i(\theta)=(X_i(\theta)-\beta_i)^+,\quad e_i(\theta)=u(X_i(\theta)).$$
Moreover, there exist cutoff abilities $0<\theta_i^l\ls\theta_i^m\ls\theta_i^h<\hi{\theta}$ such that, away from the cutoff types:
$$X_i(\theta)\lt\{\begin{array}{ll}
=0 & \mbox{ if }\theta<\theta_i^l,\\
\in(0,\beta_i] & \mbox{ if }\theta_i^l<\theta<\theta_i^m,\\
\in(\beta_i,\beta) & \mbox{ if }\theta_i^m<\theta<\theta_i^h,\\
=\beta & \mbox{ if }\theta>\theta_i^h.\end{array}\rt.$$
\end{proposition}

\begin{figure}[t]
   \begin{center}
    \subfigure[Attention]
    {\includegraphics[width=0.48\textwidth]{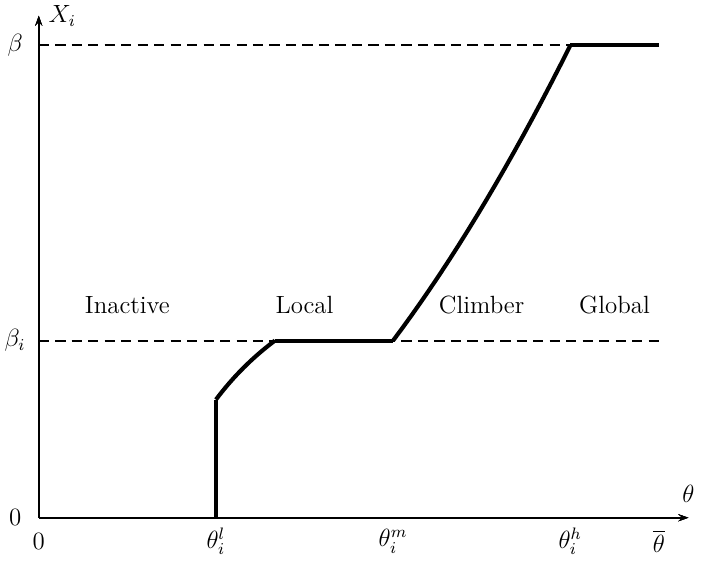}\label{fig:com-attention}}\hfill
    \subfigure[Surplus]
    {\includegraphics[width=0.48\textwidth]{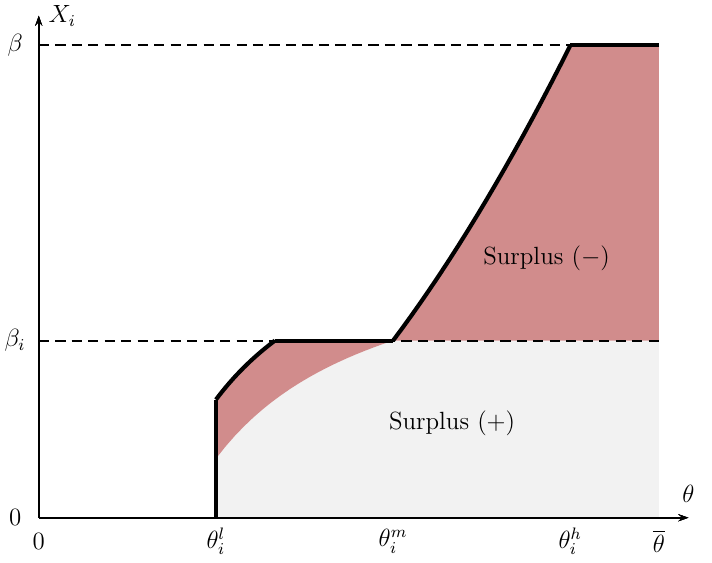}\label{fig:com-surplus}}
    \caption{Optimal Allocation and Viewer Surplus in Category $i$.}
   \end{center}
    {\scriptsize Notes: (a) The total attention $X_i$ as a function of $\theta$. When below $\beta_i$, it consists of only within-category attention. The part above $\beta_i$ is the extended reach $R_i$. (b) The gray area denotes total attention allocation if only positive-surplus attention is recommended, and the red area represents negative-surplus attention that should be recommended after accounting for network effects.}
    \label{fig:algo-com}
\end{figure}

\autoref{prop:complete} shows that the optimal algorithm can be summarized by a single increasing total-attention function, which specifies how much attention each creator type receives on the path. The function is illustrated in \autoref{fig:com-attention}. The small-cost assumption ensures that the non-negativity constraint on advertising \eqref{eq:pos} does not bind, allowing us to characterize the optimal allocation analytically.\footnote{Numerical example \ref{oa-ex4} in the Online Appendix suggests that our main findings remain robust if \eqref{eq:pos}  binds.} Given this required amount of attention, the platform sources viewers’ attention according to a \emph{pecking order}: creators with partial within-category exposure receive attention only from high-value viewers in their own categories. As creator ability rises, the algorithm expands their reach to lower-value within-category viewers. Only sufficiently high-ability creators receive cross-category reach. Regardless of the creator's ability, the effort is designed to deplete creator surplus due to the binding participation constraint. The resulting allocation divides creators into four groups: \emph{Inactive}, \emph{Local Entertainers}, \emph{Ladder Climbers}, and \emph{Global Influencers} (\autoref{fig:com-attention}).

\paragraph{Creator Inclusion} The optimal algorithm does not recommend content from creators with abilities below $\theta_i^l$. For these creators, the quality improvements induced by additional exposure are too small to justify diverting viewer attention from advertising. 
In particular, a creator of ability $\theta$ is worth including only if there exists some attention level $X_i$ such that
\begin{equation}\label{cutoff:mb-mc}
\theta u(X_i)\min\{X_i,\beta_i\}\mathbb{E}_i\lt[v\big|G_i(v)\gs 1-\tfrac{\min\{X_i,\beta_i\}}{\beta_i}\rt]\gs cX_i.
\end{equation}
The left-hand side captures the total viewer surplus generated by allocating $X_i$ units of attention to this creator: attention induces effort, which improves content quality and benefits within-category viewers. The term $\min\{X_i,\beta_i\}$ reflects that at most $\beta_i$ units of attention can come from within-category viewers; any additional attention comes from other categories and contributes to production incentives without generating direct consumption benefits.  The right-hand side is the corresponding opportunity cost, namely, the monetization value of the same attention elsewhere on the platform. The cutoff $\theta_i^l$ is therefore the lowest ability for which this condition holds for some $X_i\in(0,\beta]$. The inequality also makes clear that the cutoff $\theta_i^l$ is lower when category $i$ has a larger viewer base $\beta_i$, or when the valuation distribution $G_i$ becomes more favorable. Intuitively, when a category contains a larger mass of high-value viewers, the creator's effort is more effective at generating viewer surplus and thus profits for the platform.

\paragraph{Attention Cliff}  
The total attention $X_i(\theta)$ exhibits \emph{discontinuity} 
at $\theta_i^l$: if a creator receives attention, it must be bounded away from zero. This jump reflects the complementarity between creator effort and viewer attention within the same category: recommending content to an infinitesimal mass of viewers would generate vanishing utility for the creator and hence induce vanishing effort, and the resulting low quality cannot justify viewers' attention. Formally, in \eqref{cutoff:mb-mc}, the benefit from including a $\theta$-creator is $o(X_i)$, whereas the opportunity cost is $O(X_i)$. Hence, for sufficiently small $X_i$, the benefit is dominated by the cost, and the total attention must be bounded away from zero.

\paragraph{Attention Pecking Order}  
Viewer attention differs in its opportunity cost as an incentive instrument. Fix a creator, high-valuation viewers in his own category gain the most from his content, lower-valuation viewers in the same category gain less, and cross-category viewers receive no direct consumption value. The optimal algorithm therefore grants attention according to this ranking.

Specifically, a category-$i$ creator with ability $\theta\ls\theta_i^m$ receives attention $X_{i}(\theta)\ls \beta_{i}$; thus, it suffices to allocate attention only from high-valuation viewers within his own category, i.e., $G_i(v)\gs 1-\tfrac{X_i(\theta)}{\beta_i}$. As ability rises, the creator's audience expands to include progressively lower-valuation viewers in that category. This generates a \emph{nested} within-category allocation: higher-ability creators reach larger within-category audiences, and higher-valuation viewers watch a larger set of same-category content.

For sufficiently high-ability creators, within-category attention alone is no longer sufficient to provide the desired incentives. The platform therefore provides additional incentives through extended reach, $R_i(\theta)=X_i(\theta)-\beta_i$, represented in \autoref{fig:com-attention} by the vertical distance between $X_i(\theta)$ and $\beta_i$. Additional exposure raises effort through activating cross-category network effects, and higher ability makes the induced effort increase especially valuable. Thus $R_i(\theta)$ increases with $\theta$ and reaches its maximum at $\theta\gs\theta_i^h$.

\paragraph{Attention Subsidy} The optimal algorithm directs viewers to pay \emph{excessive} attention to content with \emph{negative pointwise surplus}, i.e., $v\theta e_i(\theta)\mathbbm{1}\{i=j\}-c<0$. In \autoref{fig:com-surplus}, the red area shows the excessive attention that seemingly generates negative surplus. From the platform's perspective, such distortions on the viewer side function as \emph{subsidies} to creators. The optimal mechanism subsidizes almost every participating creator to leverage both within- and cross-category network effects, and those subsidies are discriminatory across creators. When a creator barely qualifies for production, positive surplus is achieved only for very high-valued within-category viewers. Yet the algorithm subsidizes him by recommending the content to some lower-valued viewers within the category to take advantage of within-category network effects. As ability increases, higher content quality reduces the need for within-category subsidies, which eventually become unnecessary. Once ability becomes sufficiently high ($\theta>\theta_i^m$), the platform further subsidizes production by utilizing cross-category network effects, and the subsidy grows with ability. Overall, the subsidy magnitude is non-monotonic in creator ability, and the subsidy source differs across creator types, reflecting variation in both the benefit of incentivizing different creators and the cost of using different viewers’ attention to do so.



\begin{corollary}\label{cor:com}
Suppose $c$ is small, $u$ is strictly concave and $u(X)\ \mathbb{E}_i\lt[v\big|G_i(v)\gs 1-\tfrac{X}{\beta_i}\rt]$ is strictly increasing for all $X\in(0,\beta_i]$. Then:
$$X_i(\theta)=\lt\{\begin{array}{ll}
0 & \mbox{ if }\theta<\theta_i^l,\\
\beta_i & \mbox{ if }\theta_i^l\ls\theta\ls\theta_i^m,\\
{u'}^{-1}\lt(\frac{c}{\theta\beta_i\mathbb{E}_i[v]}\rt) & \mbox{ if }\theta_i^m<\theta<\theta_i^h,\\
\beta & \mbox{ if }\theta\gs\theta_i^h.\end{array}\rt.$$
where $\theta_i^l=\frac{c}{u(\beta_i)\mathbb{E}_i[v]}$, $\theta_i^m=\frac{c}{\beta_iu'(\beta_i)\mathbb{E}_i[v]}$, and $\theta_i^h=\frac{c}{\beta_iu'(\beta)\mathbb{E}_i[v]}$.
\end{corollary}

The monotonicity on $u(X)\ \mathbb{E}_i\lt[v\big|G_i(v)\gs 1-\tfrac{X}{\beta_i}\rt]$ requires that as a piece of content is recommended to more viewers, the increase of creator's utility must compensate the declining \emph{average} value among viewers. Then, the marginal incentive effect of additional within-category attention always dominates its marginal cost. This generates global increasing returns with respect to within-category attention, so the optimal allocation is bang-bang: if a creator's ability exceeds the inclusion cutoff, the optimal allocation expands his exposure until all viewers in the category are included and the maximum feasible within-category network effects are activated. This \emph{all-or-nothing} structure yields a closed-form expression for $\theta_i^l$, obtained by imposing equality in condition~\eqref{cutoff:mb-mc} at $X_i=\beta_i$.

\section{Incomplete Information}\label{sec:inc}

We now return to the incomplete information case and study how screening on both sides of the market interacts with network effects. We first simplify the incentive compatibility constraints using a Myersonian approach and then characterize the resulting optimization problem. We next examine how private information constrains the platform's ability to manage network effects for advertising monetization, and the implications for user participation. Finally, we discuss how to implement the optimal direct mechanism. 

\subsection{Simplifying the Design Problem}

We begin with characterizing feasible direct mechanisms. Fix a direct mechanism $\mathscr{A}$, write:
\begin{eqnarray*}
U_i(\theta)&\equiv& u(X_i(\theta))-e_i(\theta)\\
W_j(v)&\equiv&\sum_i\alpha_i\int_{\Theta}x_{ij}(\theta,v)\Big(v\theta e_i(\theta)\mathbbm{1}\{i=j\}-c\Big)\dd F_i(\theta)-ca_j(v)
\end{eqnarray*}
as $(\theta,i)$-creator's and $(v,j)$-viewer's indirect utilities under truthful report. The standard Myersonian approach provides the necessary and sufficient conditions for a direct mechanism to be incentive compatible: 

\begin{lemma}[User IC]\label{lem:IC}\ \\ 
Creators' IC constraint, \eqref{eq:IC-p}, is satisfied if and only if for all $\theta,i$:
\vspace{-10pt}\begin{itemize}\tight
\item $q_i(\theta)=\theta e_i(\theta)$ is increasing in $\theta$, and
\item $U_i(\theta)=U_i(0)+\int_0^{\theta}\frac{q_i(\tilde{\theta})}{\tilde{\theta}^2}\dd\tilde{\theta}$.
\end{itemize}
\vspace{-10pt}
Viewers' IC constraint, \eqref{eq:IC-c}, is satisfied if and only if for all $v,j$:
\vspace{-10pt}\begin{itemize}\tight
\item $Q_j(v)\equiv\alpha_j\int_{\Theta}x_{jj}(\theta,v)q_j(\theta)\dd F_j(\theta)$ is increasing in $v$, and
\item $W_j(v)=W_j(\lo{v})+\int_{\lo{v}}^vQ_j(\tilde v)\dd \tilde{v}$.
\end{itemize}
\end{lemma}

In order for creators and viewers to truthfully report their types $\theta$ and $v$, an exact amount of information rent must be granted. This limits the platform's ability to extract effort from creators and insert ads to viewers. Also, incentive compatibility requires the equilibrium quality $q_i(\theta)$ to be increasing in $\theta$, and the quality-weighted total attention $Q_j(v)$ to be increasing in $v$.

Define the \emph{virtual types} of a $(v,j)$-viewer and a $(\theta,i)$-creator as 
$$\Psi_j(v)\equiv v-\frac{1-G_j(v)}{g_j(v)},\quad \Phi_i(\theta)\equiv\theta-\frac{1-F_i(\theta)}{f_i(\theta)}$$
respectively. Following the standard procedure, we set the lowest type user's payoff to be zero for optimality. By \autoref{lem:IC}, all users' payoffs are monotone in their vertical types, so all user types' IR are automatically satisfied. Using the users' IC conditions,  
we can replace $a_j(v)$ in the objective and simplify the platform's problem to the following:
\begin{align}
\max_{\substack{\{e_i(\cdot)\}_i,\\ \{x_{ij}(\cdot,\cdot)\}_{i,j}}}&\sum_{i,j}\alpha_i\beta_j\int_{\Theta}\int_Vx_{ij}(\theta,v)\Big(\underbrace{\Psi_j(v)\theta e_i(\theta)\mathbbm{1}\{i=j\}-c}_{\mbox{\footnotesize Virtual surplus}}\Big)\dd G_j(v)\dd F_i(\theta),\label{eq:obj-inc}\\ 
\mbox{s.t. }
&U_i(\theta)=\int_0^{\theta}\frac{e_i(\tilde{\theta})}{\tilde{\theta}}\dd\tilde{\theta},\ \forall\ \theta,i,\quad \theta e_i(\theta)\mbox{ is increasing},\label{eq:U-inc}\tag{IC-I}\\
&\sum_i\alpha_i\int_{\Theta}\Bigg(x_{ij}(\theta,v)\big(v\theta e_i(\theta)\mathbbm{1}\{i=j\}-c\big)\notag\\
&\hspace{4em}-\theta e_i(\theta)\mathbbm{1}\{i=j\}\int_{\lo{v}}^vx_{ij}(\theta,\tilde{v})\dd \tilde{v}\Bigg)\dd F_i(\theta)\gs0,\ \forall\ v,j.\label{eq:pos-inc}\tag{NN-I}
\end{align}

Comparing the objective of problem \eqref{eq:obj-inc} to its counterpart \eqref{eq:obj} with the complete information, 
the viewer type $v$ in the parentheses of the objective in problem \eqref{eq:obj-inc} is now replaced by its virtual counterpart $\Psi_j(v)$. Indeed, from the platform's perspective, the ads profit created by each match between $(\theta,i)$ and $(v,j)$ is exactly the \emph{virtual viewer surplus} in the parentheses, that is, the viewer's surplus minus the associated information rent the platform has to leave to higher-type viewers. Using \autoref{lem:IC}, creators' IC can be replaced by constraint \eqref{eq:U-inc}. In contrast, the pointwise advertising constraint \eqref{eq:pos-inc} uses actual value $v$ and subtracts the information rent explicitly, mirroring \eqref{eq:pos}.

As before, to understand the trade-off in determining attention allocation, we 
plug constraint \eqref{eq:U-inc} into the objective and take pointwise first derivative with respect to $x_{ij}(\theta, v)$. After some algebra, we decompose the marginal effect of incremental attention $x_{ij}(\theta,v)$ into three terms.

\begin{definition}[Network-Effect-Adjusted Virtual Surplus]\ \\
In the optimal algorithm, the value of $x_{ij}(\theta,v)$ is determined by the sign of
\begin{eqnarray}
&&\underbrace{\Psi_j(v)\theta e_i(\theta)\mathbbm{1}\{i=j\}-c}_{\text{\emph{Virtual Surplus}}}+\underbrace{u'(X_i(\theta))\theta\beta_i\int_Vx_{ii}(\theta,\tilde{v})\Psi_i(\tilde{v})\dd G_i(\tilde{v})}_{\text{\emph{Network Effect: Current Type}}}\nn\\
&&-\underbrace{u'(X_i(\theta))\frac{\beta_i}{f_i(\theta)}\int_{\theta}^{\hi{\theta}}\int_Vx_{ii}(\tilde{\theta},\tilde{v})\Psi_i(\tilde{v})\dd G_i(\tilde{v})\dd F_i(\tilde{\theta})}_{\text{\emph{Network Effect: Higher Types}}}.\label{eq:decom-inc}
\end{eqnarray}
In particular, $x_{ij}(\theta,v)=1$ if expression \eqref{eq:decom-inc} is positive, $x_{ij}(\theta,v)=0$ if it is negative, and any value in $[0,1]$ is optimal if it is zero.
\end{definition}

The first term represents the direct pointwise effect of assigning attention, mirroring its counterpart in \eqref{eq:decom}. Instead of extracting all viewer surplus, the platform now only captures the \emph{virtual surplus} due to information rent. The second term echoes the indirect network effect in \eqref{eq:decom}. Holding constant the information rent of the current type, the incremental attention allows the platform to extract more effort, globally benefiting all within-category viewers who watch this content. The third term is new to the incomplete information setting, as it represents the ``doubly global'' effect. An increase in effort by the current type globally raises information rent for all higher types, causing them to reduce effort; this reduction, in turn, globally affects all viewers who watch their content. The resulting aggregate impact is summarized by the double integral term. 
The second term, as in the complete-information benchmark, favors greater attention allocation; whereas the third term, arising under incomplete information, works against this force, discouraging excessive attention assignment.


\subsection{Optimal Design under Incomplete Information}


We now characterize the optimal allocation when users' vertical types are private information. Throughout the rest of the paper, we assume $\Psi_j(\lo{v})>0$. As in \autoref{sec:com}, we still assume the attention cost $c$ to be small so that the non-negativity condition \eqref{eq:pos-inc} never binds.

\begin{proposition}[Algorithm: Incomplete Information]\label{prop:incomplete}\ \\
Suppose $c$ is small, $u(X)/X$ is decreasing on $(0,\beta_i/2]$, and
$u(X)\ \mathbb{E}_i\lt[\Psi_i(v)\big|G_i(v)\gs 1-\tfrac{X}{\beta_i}\rt]$ 
is strictly increasing on this interval. Then for each category $i$, an optimal algorithm can be selected with nondecreasing total attention $X_i(\cdot)$ such that:
$$\begin{aligned}
x_{ii}(\theta,v)&=\mathbbm{1}\lt\{X_i(\theta)>0,\ G_i(v)\gs 1-\tfrac{X_i(\theta)}{\beta_i}\rt\},\quad R_i(\theta)=(X_i(\theta)-\beta_i)^+,\quad
e_i(\theta)=\int_{(0,\theta]}\frac{\tilde{\theta}}{\theta}\ \dd u\big(X_i(\tilde{\theta})\big).
\end{aligned}$$
Moreover, there exist cutoff abilities $0<\theta_i^L\ls\theta_i^M\ls\theta_i^H$, with $\theta_i^L>\theta_i^l$, $\theta_i^M>\theta_i^m$, $\theta_i^H>\theta_i^h$, such that:
$$X_i(\theta)\lt\{\begin{array}{ll}
=0 & \mbox{ if }\theta<\theta_i^L,\\
\in(0,\beta_i] & \mbox{ if }\theta_i^L<\theta<\theta_i^M,\\
\in(\beta_i,\beta) & \mbox{ if }\theta_i^M<\theta<\theta_i^H,\\
=\beta & \mbox{ if }\theta\gs\theta_i^H.\end{array}\rt.$$
At ties we select the greatest exposure; this convention also assigns exposure at $\theta_i^L$ and $\theta_i^M$. 
\end{proposition}
















The monotonicity of $u(X)/X$ and $u(X)\mathbb{E}_i[\Psi_i(v)\mid G_i(v)\gs1-X/\beta_i]$ is a local regularity assumption used in the proof's common randomization argument, which preserves creator IC. The sufficient small-$c$ bound used in the proof also permits every active creator to receive at least its full own-category audience. The resulting reduced problem has a nondecreasing optimum without ironing. The first condition requires that the creator's utility cannot grow too fast when his exposure is relatively small. The second is an incomplete-information analogue of the assumption in \autoref{cor:com}: as content reaches more viewers, the increase in creator utility must compensate the declining average virtual valuation. The proof specifies separate sufficient small-$c$ requirements for this reduction, interior upper cutoffs, and advertising feasibility.

Similar to the complete information case, the creators in each category are partitioned into four intervals based on ability: Inactive, Local Entertainers, Ladder Climbers, and Global Influencers, with different thresholds. 
Within-category attention jumps from zero to a strictly positive level at $\theta_i^L$, forming an attention cliff, and is nondecreasing thereafter. Within a given category, attention allocation is monotone in viewer value: higher-value viewers consume everything that lower-value viewers do, and possibly more. Cross-category attention will be assigned only if the creator ability is significantly high ($\theta>\theta_i^M$), and is maxed out at $\theta_i^H$. Except for creators at the entry cutoff, all participating creators earn information rents; accordingly, recommended effort is distorted away from the level that would fully extract their surplus.

Private information on both sides makes it more costly for the platform to expand and monetize network effects. As a result, the optimal algorithm features both lower viewer attention and lower creator effort. Given the complementarity between attention and effort, this joint contraction is natural, but it is driven by two layers of forces. First, standard informational distortions require leaving information rents to both viewers and creators, limiting ads extraction and weakening effort incentives. Second, private information interacts with network effects: it becomes more costly to distort viewer attention and leverage network effects, while the incentive gains from such distortions are reduced. Lower effort and content quality further depress viewers' willingness to allocate attention, generating a feedback loop that dampens activity on both sides of the market. 

The interaction between network effects and private information is most clearly reflected in two distortions of the optimal allocation: 
one on the intensive margin and one on the extensive margin. 
First, network externalities imply that asymmetric information distorts allocations even for the highest types on both sides of the market. 
For a highest-valued viewer of $v=\hi{v}$, we have $\Psi_j(\hi{v})=\hi{v}$ so that normally there is no concern about leaving information rent to even higher types. However, with two-sided incomplete information her allocation is nevertheless distorted: she watches less and worse content due to information friction on the creator side. A similar logic applies to the most able creator of $\theta=\hi{\theta}$: effort is distorted downward and even attention is reduced if $\hi{\theta}$ falls below $\theta_i^H$. Although the most able creator does not have a higher type above, he still faces viewers who allocate less attention due to their own information rent, and therefore reduces effort. 

Second, private information raises the creator-inclusion cutoff $\theta_i^L$ above the complete information counterpart $\theta_i^l$. 
Under private information, the benefit of including a creator is now evaluated using viewers' virtual types that reflect the information rents left to the viewers. Meanwhile, the cost of including a creator includes the information rent paid to higher types of creators, lowering the quality of their content and thereby the utility of their interested viewers (this is the \textit{doubly global} effect). 
Both forces require the algorithm to be more selective, shutting down more low-ability creators. Moreover, information rents on both sides make attention subsidies
more costly, leading the platform to provide less subsidized
attention to participating creators and raising the threshold
for attaining Global Influencer status.

%





\paragraph{Category Survival and Reach} 
Information asymmetry contracts the platform along both intensive and extensive margins, and these distortions are especially pronounced for niche categories with thinner and lower-valued viewer bases: two-sided private information threatens their survival from both directions by further raising an already high inclusion threshold. The following corollary imposes additional assumptions and makes this point more transparent.

\begin{corollary}\label{cor:inc}
Suppose $c$ is small, $u$ is strictly concave and $u(X)\ \mathbb{E}_i\lt[\Psi_i(v)\big|G_i(v)\gs 1-\tfrac{X}{\beta_i}\rt]$ is strictly increasing for all $X\in(0,\beta_i]$. Then:
$$X_i(\theta)=\lt\{\begin{array}{ll}
\beta_i & \mbox{ if }\theta_i^L\ls\theta\ls\theta_i^M,\\
{u'}^{-1}\lt(\frac{c}{\Phi_i(\theta)\beta_i\mathbb{E}_i[\Psi_i(v)]}\rt) & \mbox{ if }\theta_i^M<\theta<\theta_i^H,\end{array}\rt.$$
where $\Phi_i(\theta_i^L)=\frac{c}{u(\beta_i)\mathbb{E}_i[\Psi_i(v)]}$, $\Phi_i(\theta_i^M)=\frac{c}{\beta_iu'(\beta_i)\mathbb{E}_i[\Psi_i(v)]}$, and $\Phi_i(\theta_i^H)=\frac{c}{\beta_iu'(\beta)\mathbb{E}_i[\Psi_i(v)]}$.
\end{corollary}



The strict concavity on $u$ grants us a closed-form characterization of the optimal allocation. For the inclusion cutoff, information asymmetry implies the following upward shift:
$$\theta_i^L\gs\Phi_i(\theta_i^L)=\frac{c}{u(\beta_i)\mathbb{E}_i[\Psi_i(v)]}>\frac{c}{u(\beta_i)\mathbb{E}_i[v]}=\theta_i^l,$$
where $\theta_i^l$ is category-$i$'s inclusion cutoff under complete information derived in Corollary \ref{cor:com}.  
The first inequality holds by definition of virtual type, capturing the increased incentive costs of creators due to private information. The second inequality arises because viewers' private information reduces the platform's ability to profit from attention on ads, reflected in $\Psi_i(v)<v$ for almost all $v$. Information asymmetry on both sides jointly raises the participation cutoff. The same logic applies to the remaining cutoffs on the production side. Notice that these cutoffs are determined by group characteristics of the same category and are shaped by information asymmetry on both sides. 

The closed-form expression above yields a simple survival condition for a category: some creator in the category can be included only if the inclusion cutoff is no larger than the highest creator ability, or
$$\mathbb{E}_i[\Psi_i(v)]u(\beta_i)\gs \frac{c}{\hi{\theta}}.$$
For each category, this formula implicitly gives the \emph{critical mass} of viewers required for the category to remain active, taking as given its value distribution and informational frictions. Importantly, the critical mass is only affected by viewer-side private information through $\mathbb{E}_i[\Psi_i(v)]$, but not by creator-side private information, since the top creator's virtual type equals his true type. For categories with viewer base $\beta_i$ such that $u(\beta_i)\in\lt(\tfrac{c}{\hi{\theta}\mathbb{E}_i[v]},\tfrac{c}{\hi{\theta}\mathbb{E}_i[\Psi_i(v)]}\rt),$ viewer-side private information can make these categories extinct, even though they would remain viable under complete information. 

The above discussion further clarifies the role of viewer-base expansion: by bringing marginal categories above their survival thresholds, a larger viewer base can mitigate the loss of network-effect leverage caused by agency problems and sustain greater category variety.

\begin{figure}[t]
    \begin{center}
    \subfigure[Quality]{\includegraphics[width=0.45\linewidth]{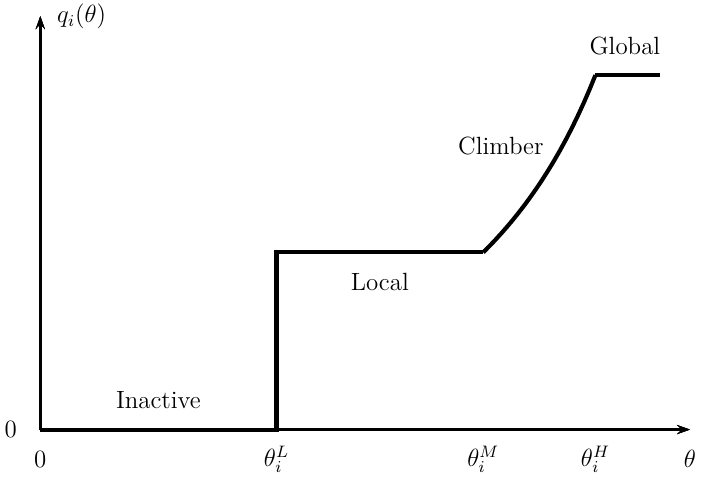}\label{fig:quality}}\hfill
    \subfigure[Effort]{\includegraphics[width=0.45\linewidth]{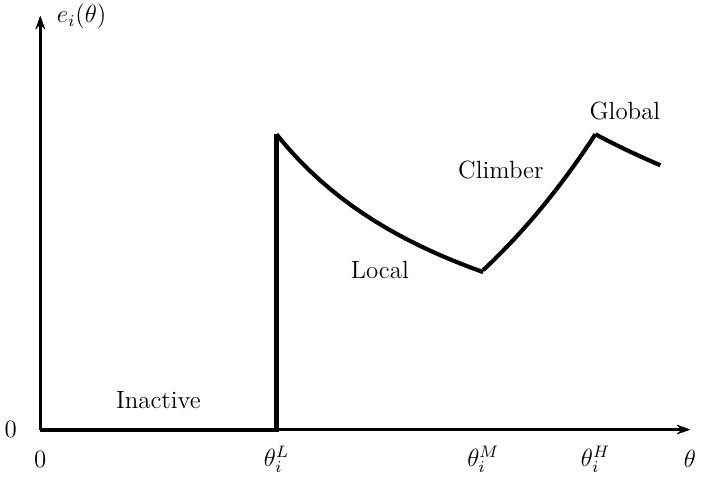}\label{fig:effort}}
    \caption{Quality and Effort as functions of ability $\theta$ for category-$i$ creators}
     \end{center}
    \label{fig:inc}
\end{figure}
\paragraph{Clustering of Content Quality} 
The interaction of network externalities and asymmetric information can induce creators with different abilities to produce similar-quality content within some groups. The resulting attention cliff compresses variation in attention rewards among Local Entertainers. When creator ability is private information, incentive compatibility translates this compression into (local) \emph{clustering of content quality}. 
When the within-category allocation is all-or-nothing, this quality clustering result becomes more stark: all Local Entertainers fall into this fixed-reward region. A similar logic applies to sufficiently able creators ($\theta\gs \theta_i^H$). Since they already reach all viewers, the platform runs out of incentive instruments and the resulting quality bunches at the maximum level. By contrast, Ladder Climbers face active trade-offs: they can marginally increase quality in exchange for marginally further reach among cross-category viewers. A typical quality function is depicted in \autoref{fig:quality}, where bunching appears as flat segments for Local Entertainers and Global Influencers.

Non-decreasing quality does not necessarily translate into non-decreasing efforts. In the flat-quality segments in \autoref{fig:quality}, more able creators exert less effort to sustain the same quality level, whereas for Ladder Climbers, effort typically increases with ability. Intuitively, 
Local Entertainers and Global Influencers face insufficient marginal return from quality improvement and thus slack off, while Ladder Climbers aspire to increase their extended reach by working harder. 
Overall, effort is non-monotone in ability and typically double peaked.

\subsection{Implementation}\label{sec:implementation}

\paragraph{Creator Side} 
The optimal direct mechanism specifies effort and total attention $(e_i(\theta), X_i(\theta))$ for each creator ability $\theta$ in category $i$. As an outcome, a $(\theta,i)$-creator produces observable content quality at $q_i(\theta)=\theta e_i(\theta)$ and earns total attention $X_i(\theta)$. To implement this, the platform can commit to an attention allocation rule $\chi_i$ that maps observed quality $q$ into the promised total attention $\chi_i(q)$ for each creator in category $i$. Specifically, we can write $q_i(\theta)=\theta e_i(\theta)=\int_0^{\theta}\tilde{\theta}\ \dd u\lt(X_i(\tilde{\theta})\rt)$ from \autoref{prop:incomplete}. Eliminating the dependence on $\theta$, we have the following result.

\begin{proposition}[Implementation: Creators]\label{prop:impc}\ \\
For creators in category $i$, the platform promises total attention $\chi_i(q)$ implicitly defined by:
\begin{eqnarray}
q=\int_{(0,\chi_i(q)]}X_i^{-1}(X)\dd u(X),\qquad 0\ls q\ls\hi{q}_i\equiv q_i(\hi{\theta}),\nn
\end{eqnarray}
where $X_i^{-1}(X)\equiv\inf\{\theta:X_i(\theta)\gs X\}$ is the generalized inverse of a right-continuous optimal exposure schedule. For $q>\hi{q}_i$, set $\chi_i(q)=X_i(\hi{\theta})$. Then effort $e_i(\theta)$ in \autoref{prop:incomplete} is optimal for $\theta$-creators.
\end{proposition}

\autoref{prop:impc} proposes one implementation directly from the optimal algorithm. Each creator maximizes $u(\chi_i(q))-\tfrac{q}{\theta}$ and self-sorts into the intended quality level for their ability. The implementation can be further simplified for practice, by truncating quality levels that are not chosen by any creator. For example, each category can set a minimum quality threshold $\lo{q}_i\equiv q_i(\theta_i^L)$, below which no attention is granted. This truncation is optional: under the untruncated rule, a creator below the entry cutoff earns $q(1/\theta_i^L-1/\theta)<0$ from any positive quality below $\lo{q}_i$. Above the highest equilibrium quality $\hi{q}_i$, attention remains at $X_i(\hi{\theta})$, which equals $\beta$ whenever the full-exposure region is reached.

\paragraph{Viewer Side}

The optimal direct mechanism assigns a content-advertising mix to each $(v,j)$-viewer. Define
\[
\vartheta_j(v)\equiv\{\theta:x_{jj}(\theta,v)=1\}=\lt\{\theta:X_j(\theta)>0,\ G_j(v)\gs1-\tfrac{X_j(\theta)}{\beta_j}\rt\}
\]
as the set of creator types in category $j$ whose content is recommended to $(v,j)$-viewers. A $(v,j)$-viewer receives all content from $\vartheta_j(v)$, together with
\[
\frac{\alpha_j}{c}\int_{\vartheta_j(v)}(v\theta e_j(\theta)-c)\dd F_j(\theta)-\frac{W_j(v)}{c}
\] 
units of non-preferred content and ads.

For implementation, we let the platform commit to a single \emph{ordered feed} for all viewers in each category, allowing proportional mixing of content and ads within a layer, as in the recommendation intensities of the direct mechanism. Each viewer consumes these layers sequentially and chooses when to stop. Although our model is static, the feed can be interpreted as unfolding over time, with each unit of viewing consuming one unit of attention.

This within-category common-feed structure creates an implementation problem absent from implementations in standard screening problems. For example, in nonlinear pricing, consumers select their preferred bundles directly from a menu \citep[Section~2.3]{borgers2015introduction}. Here, all viewers in a category face the same ordered feed and must consume everything preceding their chosen stopping point. Even when the optimal direct mechanism assigns only finitely many distinct bundles, the feed must specify what viewers consume if they stop between the prescribed attention levels or continue beyond them. Incentive compatibility among the prescribed allocations alone does not resolve these questions.

To construct the feed, define the viewer's total attention as
\[
A_j(v)\equiv\frac{vQ_j(v)-W_j(v)}{c},
\]
where $Q_j(v)=\alpha_j\int_{\vartheta_j(v)}\theta e_j(\theta)\dd F_j(\theta)$ is the cumulative within-category quality-weighted content received by a $(v,j)$-viewer. The optimal direct mechanism thus specifies a collection of points $(A_j(v),Q_j(v))$. We construct a common feed by extending these prescribed allocations to all possible stopping points. The following proposition characterizes this extension.

\begin{figure}[t]
\begin{center}
    \includegraphics[width=0.55\textwidth]{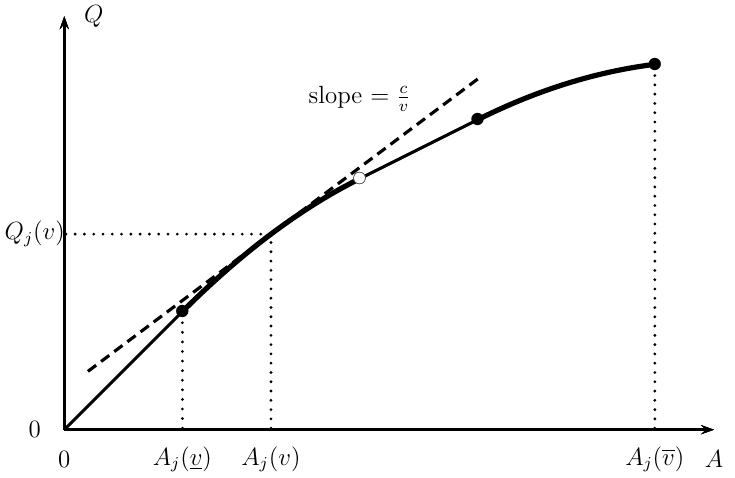}
\end{center}
\vspace{-0.5cm}
\caption{Implementation: Viewer Side}
\label{fig:allocation-ben}
\medskip
{\scriptsize Notes: The thick curves depict prescribed allocations $\{(A_j(v),Q_j(v))\}_{v\in V}$ on the concave frontier. Thin segments interpolate across unassigned attention levels, including those below the lowest prescribed stopping point. The dashed line is a supporting indifference line of a $v$-viewer in the $(A,Q)$ space, with slope $c/v$.}
\end{figure}

\begin{proposition}[Implementation: Viewers]\label{prop:impv}\ \\
The allocation in \autoref{prop:incomplete} is implemented, allowing proportional mixing within layers, by the same ordered feed for all viewers in category $j$. Cumulative attention $A$ delivers cumulative within-category quality-weighted content $K_j(A)$, where
\[
K_j(A)\equiv\sup\lt\{\sum_{k=1}^n\rho_kQ_j(v_k):n\in\mathbb{N},\ v_k\in V,\ \sum_{k=1}^n\rho_kA_j(v_k)=A,\,\rho_k\gs0,\,\sum_{k=1}^n\rho_k\ls1\rt\}
\]
is the concave envelope of $\{(0,0)\}\cup\{(A_j(v),Q_j(v))\}_{v\in V}$.
\end{proposition}

The construction fills in attention levels left unspecified by the direct mechanism. The prescribed allocations already lie on the concave frontier, as follows from viewer IC. Across any gap between prescribed stopping points, including the gap from zero to the lowest prescribed stopping point, the platform mixes the additional content and ads proportionally, so cumulative within-category quality follows the linear interpolation $K_j(A)$. Such intermediate combinations were not assigned to any type by the direct mechanism. The general construction in Figure \ref{fig:allocation-ben} allows distinct prescribed bundles; when all participating creators exhaust their own audiences, the construction reduces to a line from zero to the common prescribed bundle, followed by constant cumulative quality.

Crucially, these additional stopping options preserve incentive compatibility. Each $(v,j)$-viewer chooses $A$ to maximize $vK_j(A)-cA$. The incentive constraints of the direct mechanism ensure that stopping at $A_j(v)$ is weakly preferred to stopping at any other prescribed attention level $A_j(v')$, as well as to not watching. Since utility is linear in $(Q,A)$, no convex combination of these allocations can provide a higher payoff. Consequently, the concave envelope introduces no profitable intermediate stopping points, and each viewer optimally stops at the intended $A_j(v)$. Figure \ref{fig:allocation-ben} illustrates the construction.

The feed can also extend beyond the largest prescribed attention level $A_j(\hi{v})$. Beyond this point, the platform displays only cross-category content or ads, so that cumulative within-category quality remains constant at $Q_j(\hi{v})$. Further viewing then increases attention costs without providing additional within-category quality, and no viewer wishes to continue.

This implementation resembles the ``infinite scroll'' interface common on social media platforms. All viewers in a category encounter the same ordered feed and choose when to stop; different stopping points arise when their prescribed bundles differ. The marginal within-category quality delivered per unit of attention weakly declines along the feed, allowing the platform to sort viewers by their willingness to continue watching.

\section{Viewer Welfare Maximization}\label{sec:welfare}

In this section, we study the viewer-optimal allocation as a welfare
benchmark for evaluating ad-funded market provision of content. The
benchmark maximizes aggregate viewer utility without advertising, subject
to the same incentive and participation constraints. The central issue is whether converting viewer
surplus into advertising revenue changes content production and consumption.
Full surplus extraction can align the two objectives; imperfect extraction
can separate them. In our model, viewers' private information supplies this
monetization friction, making the platform's return from creator exposure
differ from its contribution to viewer surplus.

Under complete information, creator participation can be binding without reducing viewer welfare,
so $e_i(\theta)=u(X_i(\theta))$. The planner's problem becomes
\eqref{eq:obj}, subject to viewer participation \eqref{eq:pos}, with
$a_j(v)=0$. Its network-effect-adjusted surplus coincides with
\eqref{eq:decom}.

\begin{proposition}[Viewer-Optimal: Complete Information]\label{prop:welfare-com}\ \\
Suppose $c$ is sufficiently small. The allocation in \autoref{prop:complete} maximizes viewer surplus when all advertising is removed.
\end{proposition}

Without screening frictions, the platform maximizes surplus from content production and consumption, extracting it through type-dependent advertising loads. The viewer-optimal mechanism leaves this surplus with viewers but uses the same attention subsidies: recommendations with negative direct consumption surplus can still be desirable for their production benefits.

We now study monetization in the private-information environment of
Section~\ref{sec:inc}. Removing advertising leaves the planner subject to
the same creator and viewer IC and IR:
\begin{eqnarray}
\max_{\{e_i,x_{ij}\}}&&\sum_{i,j}\alpha_i\beta_j\int_\Theta\int_V x_{ij}(\theta,v)\Big(v\theta e_i(\theta)\mathbbm{1}\{i=j\}-c\Big)\dd G_j(v)\dd F_i(\theta)\\
\mbox{s.t.}&&\eqref{eq:IC-p}-\eqref{eq:IR-c},\quad a_j(v)=0.
\end{eqnarray}
Substituting creator IC into the objective gives the marginal surplus at a given allocation:
\begin{eqnarray}
&&\underbrace{v\theta e_i(\theta)\mathbbm{1}\{i=j\}-c}_{\text{Surplus}}+\underbrace{u'(X_i(\theta))\theta\beta_i\int_Vx_{ii}(\theta,\tilde{v})\tilde{v}\dd G_i(\tilde{v})}_{\text{Network Effect: Current Type}}\nn\\
&&-\underbrace{u'(X_i(\theta))\frac{\beta_i}{f_i(\theta)}\int_{\theta}^{\hi{\theta}}\int_Vx_{ii}(\tilde{\theta},\tilde{v})\tilde{v}\dd G_i(\tilde{v})\dd F_i(\tilde{\theta})}_{\text{Network Effect: Higher Types}}.
\label{eq:viewer-marginal}
\end{eqnarray}
Relative to \eqref{eq:decom-inc}, actual valuations replace virtual valuations
in all three terms. The planner values the viewer benefits of better content,
including the information rents that the platform cannot extract through
advertising. The monetization requirement thus changes the valuation of
creator exposure, even though both mechanisms face the same private
information. Both still use attention subsidies to stimulate production.
Within-category exposure also generates consumption benefits, whereas
cross-category exposure provides only production incentives. The following
proposition establishes the resulting creator segmentation.

\begin{proposition}[Viewer-Optimal: Private Information]\label{prop:welfare}\ \\
Suppose $u$ is continuously differentiable on $[0,\beta]$ and $u(x)/x$
is nonincreasing in a neighborhood of zero. For sufficiently small $c$,
there exists a viewer-optimal mechanism with nondecreasing exposure and
three cutoffs
$$0<\theta_i^{\mathcal L}\ls\theta_i^{\mathcal M}\ls\theta_i^{\mathcal H}$$
in each category. Creators below $\theta_i^{\mathcal L}$ are inactive; those between $\theta_i^{\mathcal L}$ and $\theta_i^{\mathcal M}$ receive only their full own-category audience; those between $\theta_i^{\mathcal M}$ and $\theta_i^{\mathcal H}$ additionally receive cross-category exposure; and those at or above $\theta_i^{\mathcal H}$ reach all viewers. 
\end{proposition}

\autoref{prop:welfare} yields a creator-type segmentation similar to
that in \autoref{prop:incomplete}. 
Under strict concavity of $u$, we obtain a sharp comparison
of the cutoffs.

\begin{proposition}[Distortions]\label{prop:welfare-comparison}\ \\
Suppose the conditions of \autoref{prop:welfare} hold and $u$ is strictly concave. The viewer-optimal cutoffs are lower than their profit-maximization counterparts:
$$\theta_i^{\mathcal L}<\theta_i^L,\quad \theta_i^{\mathcal M}<\theta_i^M,\quad \theta_i^{\mathcal H}<\theta_i^H.$$
\end{proposition}

Under these conditions, market provision excludes additional low-ability
creators and delays cross-category and full exposure relative to the
viewer-optimal benchmark. Exposure is weakly lower for every creator type.
Aggregate quality is also weakly lower in each category: creator IC gives
$\mathbb E_i[\theta e_i(\theta)]=
\mathbb E_i[\Phi_i(\theta)u(X_i(\theta))]$, and exposure differences occur
only at positive virtual abilities. The additional conditions below give explicit formulas for the cutoffs.

\begin{corollary}\label{cor:welfare-sol} Suppose the conditions of \autoref{cor:inc} hold. Then:
$$X_i(\theta)=\begin{cases}
  \beta_i,&\theta_i^{\mathcal L}\ls\theta\ls\theta_i^{\mathcal M},\\
  (u')^{-1}\lt(\dfrac{c}{\beta_i\mathbb{E}_i[v]\Phi_i(\theta)}\rt),&\theta_i^{\mathcal M}<\theta<\theta_i^{\mathcal H}.
\end{cases}$$
The three cutoffs satisfy
$$\Phi_i(\theta_i^{\mathcal L})=\frac{c}{\mathbb E_i[v]u(\beta_i)},\quad\Phi_i(\theta_i^{\mathcal M})=\frac{c}{\beta_i\mathbb E_i[v]u'(\beta_i)},\quad\Phi_i(\theta_i^{\mathcal H})=\frac{c}{\beta_i\mathbb E_i[v]u'(\beta)}.$$
\end{corollary}

The explicit cutoffs identify which categories survive. Whenever both implementations satisfy viewer IR, a positive mass of category-$i$ creators participates under viewer-optimal provision exactly when
$
\hi{\theta}\mathbb{E}_i[v]u(\beta_i)>c,
$ 
compared with $\hi{\theta}\lo{v}u(\beta_i)>c$ under ad-funded provision. Thus the monetization requirement can affect content variety as well as exposure, with a stricter survival condition for categories with small audiences. These participation and exposure comparisons do not rank each creator type's quality: greater exposure also raises information rents.

\section{Platform Payments}\label{sec:money}
 
Thus far, we have assumed that the platform incentivizes content production solely through recommendation-based attention. In practice, many platforms offer monetary rewards through programs such as TikTok Creator Rewards and the YouTube Partner Program, with substantial payments often concentrated among leading influencers. This raises a natural question: if the platform can reward creators with money, is it still profitable to distort recommendations to provide attention subsidies?
In this section, we return to the advertising-funded platform and allow it to use monetary transfers alongside recommendation-based attention. We show that monetary payments can complement, rather than simply replace, algorithmic attention: attention subsidies can remain optimal even when transfers are feasible, and the platform may combine the two instruments to motivate content production. Moreover, under mild conditions, monetary payments are optimally reserved for the highest-ability creators. 

To formally incorporate the monetary channel, we assume that a creator values both attention and money. His gross reward utility from total attention $X\gs0$ and transfer $T\gs0$ is $u(Y)$, where $Y\equiv\pi(X)+T$ is the \emph{total money-metric reward}, decomposed into $\pi(X)$, the money-metric payoff from attention $X$, and direct payment $T$. Net utility is $u(Y)-e$. We assume $u$ is strictly concave to capture overall diminishing returns while allowing $\pi$ to be any increasing and continuously differentiable function with $\pi(0)=0$. We assume that the transfer problem has a finite optimum.\footnote{For example, $\lim_{Y\to\infty}u^{\prime}(Y)=0$ is sufficient for optimal monetary rewards to be bounded: abilities and viewer values are bounded, whereas transfers have a strictly positive marginal cost.}

With monetary transfers, the direct mechanism receives report $\theta$ and recommends effort $e_i(\theta)$. If the observable quality $q$ matches $\theta e_i(\theta)$, then attention profile $x_{i\cdot}(\theta,\cdot)$ is assigned and a direct transfer $T_i(\theta)\gs0$ is made.\footnote{We focus on nonnegative transfers. Allowing $T_i(\theta)<0$ would amount to allowing the platform to sell recognition or visibility to creators, as in models of paid certification or status provision; see, e.g., \cite{BarIsaacDebMitchell2025}.} As in the main model, we substitute out the platform-inserted ads $a_j(v)$ and simplify the platform's problem to:
\begin{eqnarray}
\hspace*{-16pt}\max_{\substack{\{e_i(\cdot)\}_i,\{T_i(\cdot)\}_i,\\ \{x_{ij}(\cdot,\cdot)\}_{i,j}}}&\hspace*{-6pt}&\hspace*{-6pt}\sum_i\alpha_i\int_{\Theta}\lt(\sum_j\beta_j\int_Vx_{ij}(\theta,v)\Big(\Psi_j(v)\theta e_i(\theta)\mathbbm{1}\{i=j\}-c\Big)\dd G_j(v)-\frac{c}{z}T_i(\theta)\rt)\dd F_i(\theta),\nn\\
\hspace*{-16pt}\mbox{s.t.}&\hspace*{-6pt}&\hspace*{-6pt}e_i(\theta)+U_i(\theta)=u\lt(\pi(X_i(\theta))+T_i(\theta)\rt),\ \forall\ \theta,i,\nn\\
\hspace*{-16pt}&\hspace*{-6pt}&\hspace*{-6pt}U_i(\theta)=\int_0^{\theta}\frac{e_i(\tilde{\theta})}{\tilde{\theta}}\dd\tilde{\theta},\ \forall\ \theta,i,\quad \theta e_i(\theta) \mbox{ is increasing}.\nn
\end{eqnarray}
The incentives enter the problem in two places. On one hand, $T_i(\theta)$ is directly deducted from the platform's objective. On the other hand, it enters creators' utility so that higher effort can be extracted. We solve the relaxation omitting \eqref{eq:pos-inc}, while retaining creator IC and nonnegative effort. The resulting allocation solves the original problem whenever its recommendations and the viewer envelope imply nonnegative advertising; a sufficient check is $\lo{v}Q_j(\lo{v})\gs c\alpha$ for every $j$.

In the rest of this section, maintain sufficiently small $c$. We assume that $\pi'(X)>z$ for $X\in[0,\beta_i]$, $\tfrac{u(\pi(X))}{X}$ is decreasing and $u(\pi(X))\mathbb{E}_i\lt[\Psi_i(v)\mid G_i(v)\gs1-\tfrac{X}{\beta_i}\rt]$ is strictly increasing for $X\in(0,\tfrac{\beta_i}{2}]$. 

\begin{proposition}[Optimal Algorithm with Money]\label{prop:money}\ \\
For each category $i$, the optimal algorithm is characterized by nondecreasing total reward $Y_i(\theta)=\pi(X_i(\theta))+T_i(\theta)$ and nondecreasing attention $X_i(\theta)$, with $x_{ii}(\theta,v)$ and $R_i(\theta)$ defined as in \autoref{prop:incomplete}, and a transfer $T_i(\theta)$. Effort satisfies $e_i(\theta)=\frac{1}{\theta}\int_{(0,\theta]}\tilde{\theta}\ \dd u\lt(\pi(X_i(\tilde{\theta}))+T_i(\tilde{\theta})\rt)$. No money is paid when $X_i(\theta)\ls\beta_i$. 
For $\theta$ such that $X_i(\theta)\gs\beta_i$, the ability space is partitioned into alternating intervals with the following patterns:
\begin{itemize}\tight
    \item $X_i(\theta)$ increases and $T_i(\theta)=0$;
    \item $X_i(\theta)$ is constant and $T_i(\theta)$ increases;
    \item both $X_i(\theta)$ and $T_i(\theta)$ are constant. 
\end{itemize}
If $\pi$ is concave, 
there exists $\theta_i^{\$}\gs\theta_i^M$ such that $T_i(\theta)=0$ for $\theta\ls\theta_i^{\$}$, and $T_i(\theta)>0$ and $X_i(\theta)=X_i(\theta_i^{\$})$ for $\theta>\theta_i^{\$}$.
\end{proposition}

\autoref{prop:money} predicts a similar segmentation of creator abilities by the total attention they receive. It shows how monetary transfers interact with attention allocation on the extended-reach segment, where $X_i(\theta)\gs\beta_i$.\footnote{Under $\pi'(X)>z$ on $[0,\beta_i]$, creators receiving only within-category attention are rewarded through exposure rather than money. Own-category attention both rewards creators and generates viewer consumption surplus.} As promised reward varies, transfers increase while exposure is fixed and can reset to zero when exposure jumps to a new contact point of the resource frontier. This description is in reward space: an optimal reward schedule can jump across a contact point or several regimes, so an exposure jump in ability need not leave the transfer equal to zero immediately afterwards. When $\pi$ is concave, positive payments, if used, are confined to the highest-ability creators. Exposure plateaus and monetary payments provide additional incentives, with possible pooling before cash begins.

We briefly discuss the intuition here. To incentivize these creators, what matters is the total money-metric reward $Y_i=\pi(X_i)+T_i$. Once the mechanism determines a reward $Y\gs\pi(\beta_i)$ on this segment, what remains is a cost-minimization problem:
\begin{equation}\label{min-payment}
\min_{\beta\gs X_i\gs\beta_i,T_i\gs 0}\; X_i+T_i/z
\quad\text{s.t.}\quad
\pi(X_i)+T_i\gs Y.
\end{equation}
The objective is the \emph{attention-equivalent cost} to the platform. Direct exposure costs attention one-for-one, while one unit of money costs $1/z$ units of advertising attention. Thus, the platform delivers the required reward using extended reach, money, or both, depending on which instrument provides the cheaper reward. 


To characterize how the platform optimally resolves the trade-off between exposure and transfers, consider the dual of problem \eqref{min-payment}:
\begin{equation*}
\hat\pi(\hat{X})\equiv \max_{\beta\gs X\gs\beta_i,T\gs0}\pi(X)+T\quad\mbox{s.t. }\quad X+\frac{T}{z}\ls\hat{X}.
\end{equation*}
where $\hat X\gs\beta_i$ measures attention-equivalent resources. Actual exposure $X$ is bounded by $\beta$, whereas $\hat X=X+T/z$ has no upper bound. One unit of exposure uses one unit of resources, while one unit of monetary transfer uses $1/z$ units because each unit of attention can generate $z$ units of advertising revenue.

\begin{figure}[t]
\begin{center}
\includegraphics[width=0.6\textwidth]{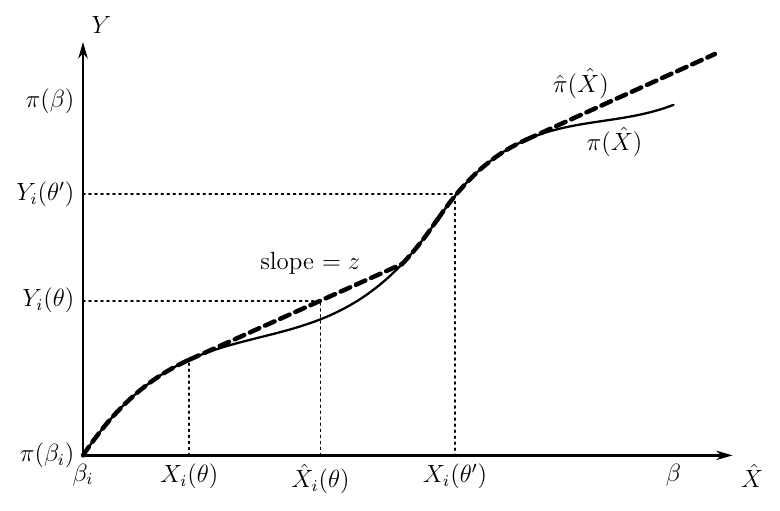}
\end{center}
\vspace{-0.5cm}
\caption{Illustration of Optimal Allocation}
\medskip
{\scriptsize Notes: The solid curve represents $\pi$, while the dashed curve represents $\hat{\pi}$. Suppose the total money-metric reward is $Y_i(\theta)$ for some $\theta$, such that $\pi(\min\{\beta,\hat{\pi}^{-1}(Y_i(\theta))\})<Y_i(\theta)$. Then the creator receives attention $X_i(\theta)$ and a monetary transfer $T_i(\theta)=Y_i(\theta)-\pi(X_i(\theta))>0$. Suppose the total money-metric reward is $Y_i(\theta')$ for some $\theta'$, such that $\pi(\min\{\beta,\hat{\pi}^{-1}(Y_i(\theta'))\})=Y_i(\theta')$. The creator receives attention $X_i(\theta')$ and zero monetary transfer.}\label{fig:allocation-ben1}
\end{figure}

This optimization problem implies two properties. First, when $\hat X\ls\beta$, the platform can spend all resources as exposure, implying
\begin{equation}\label{hat-c1}
\hat\pi(\hat{X})\gs\pi(\hat{X}), \quad \forall \beta\gs \hat{X}\gs \beta_i.
\end{equation}
Second, after reaching any point $(\hat{X},\hat\pi(\hat{X}))$, each additional unit of attention-equivalent resource can always be converted into $z$ units of monetary reward, implying
\begin{equation}\label{hat-c2}
\hat\pi(\hat{X}')-\hat\pi(\hat{X})\gs z(\hat{X}'-\hat{X}),\quad \forall \hat{X}'\gs \hat{X}\gs\beta_i.
\end{equation}
Conversely, 
the displayed frontier is the pointwise smallest function that majorizes $\pi$ on $[\beta_i,\beta]$ and whose increments have at least slope $z$. Therefore, $\hat\pi$ is the pointwise smallest function satisfying conditions \eqref{hat-c1} and \eqref{hat-c2}. Geometrically, $\hat\pi$ is the $z$-ironing of $\pi$, obtained as the upper envelope of $z$-sloped rays starting from feasible points $(X,\pi(X))$, $X\in[\beta_i,\beta]$, and pointing to the right (\autoref{fig:allocation-ben1}). Equivalently, $\hat\pi(\hat X)=z\hat X+\max_{\beta_i\ls X\ls\min\{\beta,\hat X\}}\{\pi(X)-zX\}$, so the frontier is affine with slope $z$ beyond $\beta$. It characterizes the cheapest way to deliver any reward $Y\gs\pi(\beta_i)$ in problem \eqref{min-payment}. Where $\hat\pi$ coincides with the pure-exposure curve, exposure alone is optimal. Where it lies strictly above that curve, the optimal allocation combines exposure with money. The actual transfer is $T_i=Y_i-\pi(X_i)=z(\hat X_i-X_i)$, rather than the vertical gap between the two curves at a common resource level.

The following lemma shows that the optimal allocation on the extended-reach segment is exactly the cost-minimizing implementation encoded by $\hat\pi$.

\begin{lemma}[Resource Optimization]\label{lem:resource}\ \\
When $X_i(\theta)\gs \beta_i$, we must have $Y_i(\theta)=\hat\pi(\hat X_i(\theta))$, where $\hat X_i(\theta)=X_i(\theta)+T_i(\theta)/z$ is the corresponding attention-equivalent resource being used. Given $Y_i(\theta)$, the optimal attention level is $X_i(\theta)=\max\lt\{X\in[\beta_i,\min\{\beta,\hat{\pi}^{-1}(Y_i(\theta))\}]\ \big|\ \pi(X)+z\lt(\hat\pi^{-1}(Y_i(\theta))-X\rt)=Y_i(\theta)\rt\}$. Ties are resolved in favor of greater exposure. 
\end{lemma}

\autoref{lem:resource} separates the incentive problem from the resource-allocation problem. The incentive problem determines the total money-metric reward $Y_i(\theta)$ that must be promised to each creator type. Given this reward, the resource-allocation problem determines the least costly way to implement it using extended reach and monetary transfers.

Returning to \autoref{fig:allocation-ben1}, the same geometry describes the allocation as promised reward varies. A reward on a part of $\hat\pi$ coinciding with $\pi$ is delivered by exposure alone. On a straight segment above $\pi$, exposure is held fixed and $T_i(\theta)=Y_i(\theta)-\pi(X_i(\theta))$ increases with promised reward. At a new contact with $\pi$, exposure can jump and the transfer resets to zero. This does not require the optimal reward schedule to pass continuously through every such contact: non-concavity can cause upward jumps that skip one or several resource regimes.

The ironing construction also clarifies why a local comparison between $\pi'(X)$ and $z$ need not determine when exposure resumes. With non-concave $\pi$, restarting exposure can require a discrete jump, so the relevant comparison includes its average reward gain. With concave $\pi$, the constrained maximizer of $\pi(X)-zX$ identifies the exposure at which money takes over, if the ability support reaches that reward. If exposure first reaches the physical cap $\beta$ with $\pi'(\beta)>z$, there can be an interval with $X_i(\theta)=\beta$ and $T_i(\theta)=0$ before payments start.

We next study how the platform’s reliance on exposure and monetary transfers changes with the profitability of advertising.


\begin{proposition}[Comparative Statics on Ads Income]\label{prop:z}\ \\
For a $(\theta,i)$-creator remaining on the extended-reach segment at both values of $z$, $X_i(\theta)$ weakly decreases and $T_i(\theta)$ weakly increases with $z$, within the maintained parameter range. At ties, select the greatest exposure among all optimal pairs $(X_i(\theta),T_i(\theta))$.
\end{proposition}

This comparative static follows naturally from the paper’s core trade-off in attention allocation. A higher $z$ means that each unit of viewer attention is more valuable in direct advertising monetization. As a result, diverting attention away from viewers’ preferred content in order to incentivize creators becomes more costly. The platform therefore relies more on monetary transfers and less on distorted attention allocation to provide creator incentives. 

The novelty is to link the platform’s advertising profitability to the accuracy of its recommendation algorithm. When advertising monetization is weak, attention is relatively cheap as an incentive instrument, so the platform has a stronger reason to use extended reach to reward creators. When advertising monetization is strong, attention is too valuable to waste on incentive provision, and the platform substitutes toward direct payments. Hence, platforms with greater advertising profitability should, all else equal, make more accurate recommendations. The parameter $z$ can be interpreted broadly. It may reflect the platform’s market power in advertising, the nature of the content environment, such as long videos versus short posts, or technological and legal constraints on profiling-based advertising. Restrictions that limit the platform’s ability to monetize attention through targeted ads reduce $z$, making exposure-based incentives relatively more attractive and potentially increasing recommendation distortions.

By lowering the profitability of monetizing viewer attention, such regulations make exposure-based incentives more attractive and induce platforms to distort recommendations more aggressively to motivate content creation. Importantly, this effect arises in equilibrium through the platform’s endogenous choice of incentive instruments. The model therefore provides a new perspective for policy discussions on advertising regulation and recommendation systems, highlighting a trade-off between limiting data-driven monetization and preserving recommendation accuracy.

\section{Conclusion}\label{sec:conclusion}

We build a framework to study a monopolistic social media platform's optimal content recommendation algorithm. The key trade-off is between monetizing attention and allocating attention to amplify endogenous network effects that increase content quality and, in turn, expand the revenue base. Our framework highlights the dual incentive role of recommendation systems: they shape viewers’ participation and creators’ production incentives through the allocation of attention.

Our model can be extended in many ways. One direction is to add monetary transfers between the platform and viewers. For example, YouTube offers an ad-free premium membership, allowing viewers to pay to reduce ads exposure (see, e.g., 
\cite{BiscegliaBonattiScottMorton2025}). One could allow viewers to pay  to avoid ads and even cross-category content. While such options would limit the platform's ability to leverage network effects, they would generate direct revenue, raising the question of how to optimally price the network effects a viewer generates.

A related extension is to add personalized advertising. In our baseline model, each unit of attention generates a constant advertising return. In practice, however, the value of exposure depends on the viewer's characteristics and the category of content consumed. Allowing advertising revenue to vary with viewer type and category would introduce an additional design margin: the algorithm would simultaneously manage production incentives and allocate attention to maximize heterogeneous advertising surplus. This would connect our framework to the literature on targeted advertising and data-driven ad markets (see, e.g., \cite{BergemannBonatti2011,BergemannBonatti2024}, 
and \cite{IchihashiJeonKim2026}), while embedding targeting decisions within a unified two-sided screening environment.



\singlespacing

\section*{Appendix: Proofs}

\begin{proof}[Proof of \autoref{lem:binding}]
Suppose \eqref{eq:IR-c} is slack for a positive $G_j$-measure of $v$'s in some category $j$. Then the platform should increase $a_j(v)$ for these $v$'s to increase profit.

Now suppose \eqref{eq:IR-p} is slack for a positive $F_i$-measure of $\theta$'s in some category $i$. Without loss of generality, the platform can increase $e_i(\theta)$ for these $\theta$'s to increase the quality of their content. This in turn weakly relaxes \eqref{eq:IR-c} for all viewers $(v,j)$ and thus increases profit by inserting more ads for them.
\end{proof}

\begin{proof}[Proof of \autoref{lem:sep}]
Let $x_{ij}^*(\theta,v)=\frac{\hat{R}_i(\theta)}{\beta-\beta_i}$ for all $\theta,i,v$ and $j\ne i$ so that $\sum_{j\ne i}\beta_j\int_Vx_{ij}^{*}(\theta,v)\dd G_j(v)=\hat{R}_i(\theta)$. Notice that \eqref{eq:obj} can be rewritten as:
\begin{eqnarray}
\sum_i\alpha_i\lt(\beta_i\int_{\Theta}\int_V x_{ii}(\theta,v)\Big(v\theta e_i(\theta)-c\Big)\dd G_i(v)\dd F_i(\theta)-c\int_{\Theta}R_i(\theta)\dd F_i(\theta)\rt).\nn
\end{eqnarray}
If $\{e_i^*,x_{ij}^*\}_{i,j}$ does not solve problem \eqref{eq:obj}, then there exists $i\in\mathcal{N}$, $e_i$ and $\{x_{ij}\}_j$ such that $R_i(\theta)\equiv\sum_{j\ne i}\beta_j\int_Vx_{ij}(\theta,v)\dd G_j(v)$, $e_i(\theta)=u\Big(\beta_i\int_Vx_{ii}(\theta,v)\dd G_i(v)+R_i(\theta)\Big)$, and:
\begin{eqnarray}
&&\beta_i\int_{\Theta}\int_V x_{ii}(\theta,v)\Big(v\theta e_i(\theta)-c\Big)\dd G_i(v)\dd F_i(\theta)-c\int_{\Theta}R_i(\theta)\dd F_i(\theta)\nn\\
&>&\beta_i\int_{\Theta}\int_V \hat{x}_{ii}(\theta,v)\Big(v\theta \hat{e}_i(\theta)-c\Big)\dd G_i(v)\dd F_i(\theta)-c\int_{\Theta}\hat{R}_i(\theta)\dd F_i(\theta).\nn
\end{eqnarray}
This implies that $\{\hat{e}_i,\hat{x}_{ii},\hat{R}_i\}_i$ is not a solution to problem \eqref{eq:sep}, a contradiction.
\end{proof}

\begin{proof}[Proof of \autoref{prop:complete}]
The effort choice is derived from \autoref{lem:binding}: 
$$e_i(\theta)=u\lt(\beta_i\int_Vx_{ii}(\theta,v)\dd G_i(v)+R_i(\theta)\rt).$$
Plugging \eqref{eq:eff} into \eqref{eq:obj} and taking derivative w.r.t.\ $x_{ij}(\theta,v)$, we have \eqref{eq:decom}. For $i=j$ and $X_i(\theta)>0$, then $e_i(\theta)>0$ and \eqref{eq:decom} is strictly increasing in $v$ for all $\theta>0$ and therefore $x_{ii}(\theta,\cdot)$ must be a step function. If $X_i(\theta)=0$ or $i\ne j$, \eqref{eq:decom} does not depend on $v$.

Notice that if $x_{ii}(\theta,\lo{v})=0$, then \eqref{eq:decom} $\ls0$ for $x_{ii}(\theta,\lo{v})$ implies \eqref{eq:decom} $<0$ for $x_{ij}(\theta,v)$ for any $v$ and $j\ne i$, and thus $R_i(\theta)=0$. Therefore, we use only one variable $X_i(\theta)$ to characterize optimal algorithm without ambiguity, where $X_i(\theta)\ls\beta_i$ implies $R_i(\theta)=0$ and $x_{ii}(\theta,v)=\mathbbm{1}\lt\{X_i(\theta)>0,\ G_i(v)\gs 1-\tfrac{X_i(\theta)}{\beta_i}\rt\}$, and $X_i(\theta)>\beta_i$ implies $R_i(\theta)>0$ and $x_{ii}(\theta,v)=1$ for all $v$.

Category $i$'s sub-problem simplifies to:
\begin{eqnarray}
&&\max_{X_i}\int_{\Theta}H_i^C(\theta,X_i)\dd\theta,\nn\\
\mbox{where}&&\frac{H_i^C(\theta,X_i)}{f_i(\theta)}\equiv\beta_i\theta u(X_i)\int_{G_i^{-1}(1-\tfrac{\min\{X_i,\beta_i\}}{\beta_i})}^{\hi{v}}v\dd G_i(v)-cX_i.\nn
\end{eqnarray}

Note that
$$\frac{\partial^2}{\partial\theta\partial X}\frac{H_i^C(\theta,X)}{f_i(\theta)}=\lt\{\begin{array}{ll}
\displaystyle{u(X)G_i^{-1}\lt(1-\tfrac{X}{\beta_i}\rt)+\beta_iu'(X)\int_{G_i^{-1}(1-\tfrac{X}{\beta_i})}^{\hi{v}}v\dd G_i(v)>0} & \mbox{ if }X\ls\beta_i,\\
\displaystyle{\beta_iu'(X)\mathbb{E}_i[v]>0} & \mbox{ if }X>\beta_i.\end{array}\rt.$$
According to Topkis' theorem, the greatest maximizing $X_i(\theta)$ is nondecreasing in $\theta$. Define $\theta_i^l\equiv\inf\{\theta:X_i(\theta)>0\}$, $\theta_i^m\equiv\inf\{\theta:X_i(\theta)>\beta_i\}$, and $\theta_i^h\equiv\inf\{\theta:X_i(\theta)=\beta\}$. By definition, $\theta_i^h\gs\theta_i^m\gs\theta_i^l$. For every $X>0$,
\[
\frac{H_i^C(\theta,X)}{f_i(\theta)X}\ls\theta\hi{v}u(X)-c\ls\theta\hi{v}u(\beta)-c.
\]
Consequently $X_i(\theta)=0$ whenever $\theta<c/(\hi{v}u(\beta))$, so $\theta_i^l>0$. Because $u(X)\lmt0$ as $X\downarrow0$, the first bound is negative for sufficiently small positive $X$, uniformly for $\theta$ near $\theta_i^l$. Positive maximizers near entry are therefore bounded away from zero. By continuity and the greatest maximizer convention, $X_i(\theta_i^l)>0$. Sufficiently small $c$ also places $\theta_i^h$ below $\hi{\theta}$: the benefit from exposure is strictly increasing and has a positive left derivative at $X=\beta$, so $X=\beta$ maximizes for sufficiently high types when $c$ is small.

It remains to verify \eqref{eq:pos}. Under the uniform foreign assignment in \autoref{lem:sep}, and indeed for any foreign assignment satisfying the recommendation bounds, its left-hand side is bounded below, for every $v$, by
\[
\lo{v}\alpha_j\int_{\{\theta:X_j(\theta)\gs\beta_j\}}\theta u(X_j(\theta))\dd F_j(\theta)-c\alpha.
\]
Fix $\delta\in(0,\hi{\theta})$. Since $N>1$ and $u$ is strictly increasing, at $c=0$ the payoff from $X=\beta$ exceeds the payoff from any $X\ls\beta_j$ by a strictly positive amount uniformly for $\theta\gs\delta$. Thus $X_j(\theta)\gs\beta_j$ for all $\theta\gs\delta$ when $c$ is sufficiently small. The displayed bound is then at least $\lo{v}\alpha_j\delta u(\beta_j)(1-F_j(\delta))-c\alpha>0$ for sufficiently small $c$. This verifies advertising feasibility for every viewer and completes the relaxation-and-verification argument.
\end{proof}

\begin{proof}[Proof of \autoref{cor:com}]
If the inequality holds, then for any $\theta>0$ and any $X\in(0,\beta_i)$, we have:
$$\frac{H_i^C(\theta,X)}{f_i(\theta)X}=\theta u(X)\mathbb{E}_i\lt[v\big|G_i(v)\gs 1-\tfrac{X}{\beta_i}\rt]-c.$$
By assumption, this is strictly increasing, and therefore we have $X_i(\theta_i^l)=\beta_i$. Define $\theta_i^l=\frac{c}{u(\beta_i)\mathbb{E}_i[v]}$, then $H_i^C(\theta_i^l,\beta_i)=0$.

Moreover, when $u$ is strictly concave, $\frac{\partial}{\partial X_i}\frac{H_i^C(\theta,X_i)}{f_i(\theta)}=\beta_i\mathbb{E}_i[v]\theta u'(X_i)=c$ for $X_i\gs\beta_i$. Setting $X_i(\theta_i^m)=\beta_i$ and $X_i(\theta_i^h)=\beta$ yields the thresholds.
\end{proof}

\begin{proof}[Proof of \autoref{lem:IC}]
(i) Necessity. If \eqref{eq:IC-p} holds, then for any $\theta>\theta'>0$:
\begin{eqnarray}
U_i(\theta)&\gs& U_i(\theta')+\frac{q_i(\theta')}{\theta'}-\frac{q_i(\theta')}{\theta},\nn\\
U_i(\theta')&\gs& U_i(\theta)+\frac{q_i(\theta)}{\theta}-\frac{q_i(\theta)}{\theta'}.\nn
\end{eqnarray}
Adding up, we have $(\theta-\theta')(q_i(\theta)-q_i(\theta'))\gs0$. Therefore, monotonicity holds. Taking the limit as $\theta'\lmt \theta$, we use Sandwich theorem to have $U_i'(\theta)=\frac{q_i(\theta)}{\theta^2}$, which implies the integral condition.

Sufficiency. Suppose monotonicity and integral conditions hold. Then for any $\theta,\theta'>0$:
\begin{eqnarray}
U_i(\theta)&=&U_i(\theta')+\int_{\theta'}^{\theta}\frac{q_i(\tilde{\theta})}{\tilde{\theta}^2}\dd\tilde{\theta}\gs U_i(\theta')+\int_{\theta'}^{\theta}\frac{q_i(\theta')}{\tilde{\theta}^2}\dd\tilde{\theta}\nn\\
&=&U_i(\theta')+q_i(\theta')\lt(\frac{1}{\theta'}-\frac{1}{\theta}\rt)=u\Big(\sum_j\beta_j\int_Vx_{ij}(\theta',v)\dd G_j(v)\Big)-\frac{\theta'e_i(\theta')}{\theta}.\nn
\end{eqnarray}

(ii) Necessity. If \eqref{eq:IC-c} holds, then for any $v,v'$:
\begin{eqnarray}
W_j(v)&\gs& W_j(v')-(v'-v)Q_j(v'),\nn\\
W_j(v')&\gs& W_j(v)-(v-v')Q_j(v).\nn
\end{eqnarray}
Adding up, we have $(v-v')(Q_j(v)-Q_j(v'))\gs0$. Therefore, monotonicity holds. Taking the limit as $v'\lmt v$, we use Sandwich theorem to have $W_j'(v)=Q_j(v)$, which implies the integral condition.

Sufficiency. Suppose monotonicity and integral conditions hold. Then for any $v,v'$:
\begin{eqnarray}
W_j(v)&=&W_j(v')+\int_{v'}^vQ_j(\tilde{v})\dd\tilde{v}\gs W_j(v')+\int_{v'}^vQ_j(v')\dd\tilde{v}\nn\\
&=&W_j(v')+(v-v')Q_j(v')\nn\\
&=&\sum_i\alpha_i\int_{\Theta}x_{ij}(\theta,v')\Big(v\theta e_i(\theta)\mathbbm{1}\{i=j\}-c\Big)\dd F_i(\theta)-ca_j(v').\nn
\end{eqnarray}
\end{proof}

\begin{proof}[Proof of \autoref{prop:incomplete}]
We retain creator IC and nonnegative effort throughout. By \autoref{lem:IC}, both $q_i(\theta)=\theta e_i(\theta)$ and $X_i(\theta)$ are nondecreasing, and the creator envelope gives
$$q_i(\theta)=\theta u(X_i(\theta))-\int_0^\theta u(X_i(\tilde{\theta}))\dd\tilde{\theta}=\int_{(0,\theta]}\tilde{\theta}\dd u(X_i(\tilde{\theta})).$$
For fixed exposure and nonnegative quality, assign the highest-valued own-category viewers first and use cross-category attention only after exhausting the own audience. Integrating by parts gives
$$\beta_i\int_{G_i^{-1}(1-\tfrac{\min\{X,\beta_i\}}{\beta_i})}^{\hi v}\Psi_i(v)\dd G_i(v)=\min\{X,\beta_i\}G_i^{-1}\lt(1-\tfrac{\min\{X,\beta_i\}}{\beta_i}\rt)\equiv J_i(X),$$
with $J_i(0)=0$ by continuity. Thus, category $i$ contributes $\int_\Theta[J_i(X_i(\theta))q_i(\theta)-cX_i(\theta)]\dd F_i(\theta)$ to normalized profit. The own-category recommendation rule and $R_i(\theta)=(X_i(\theta)-\beta_i)^+$ follow directly.

Apply the truncation, common randomization, and full own-category audience argument in the proof of \autoref{prop:welfare}, using the weights $\omega(v)=\Psi_i(v)$ and choosing $\hat\beta=\beta_i/2$. This argument retains creator IC throughout and uses precisely the local decrease of $u(X)/X$ and increase of $u(X)J_i(X)/X$ assumed here; it does not require differentiability of the creator density. It first removes exposure below $\Phi_i^{-1}(0)>0$, then removes exposures in $(0,\beta_i/2)$, and finally replaces every remaining positive exposure below $\beta_i$ by $\beta_i$. The last change increases normalized profit by at least
$$\lt(\Psi_i(\lo v)\Phi_i^{-1}(0)u(\beta_i/2)-c\rt)\mathbb E_i[\hat X_i-X_i].$$
Consequently, if $c<\min_i\{\Psi_i(\lo v)\Phi_i^{-1}(0)u(\beta_i/2)\}$, an optimum can be selected with $X_i(\theta)\in\{0\}\cup[\beta_i,\beta]$. This sufficient small-cost bound is separate from the bounds ensuring viewer feasibility and interior upper cutoffs. In particular, the attention cliff can be selected to reach at least the full own-category audience.

For this allocation, $J_i(X_i(\theta))q_i(\theta)=\beta_i\lo v\,q_i(\theta)$, since inactive types have zero quality. Fubini and the creators' IC therefore reduce the objective exactly to
$$\int_\Theta\lt(\beta_i\lo v\,\Phi_i(\theta)u(X_i(\theta))-cX_i(\theta)\rt)\dd F_i(\theta).$$
Choose the largest maximizer at every type:
$$X_i(\theta)\in\arg\max_{X\in\{0\}\cup[\beta_i,\beta]}\lt\{\beta_i\lo v\,\Phi_i(\theta)u(X)-cX\rt\}.$$
Because $\Phi_i$ and $u$ are increasing, this objective has increasing differences between distinct exposures. Its largest maximizer is nondecreasing, so the pointwise optimum satisfies creator IC and nonnegative effort and attains the upper bound without ironing. This establishes the asserted effort formula without relaxing either creator constraint.

Define $\theta_i^L\equiv\inf\{\theta:X_i(\theta)>0\}$, $\theta_i^M\equiv\inf\{\theta:X_i(\theta)>\beta_i\}$, and $\theta_i^H\equiv\inf\{\theta:X_i(\theta)=\beta\}$. These are ordered by monotonicity. At $\Phi_i(\theta)=0$, positive exposure has strictly negative payoff, so continuity and the lower bound on positive exposure imply $\theta_i^L>\Phi_i^{-1}(0)>0$. Right continuity under the largest-maximizer convention gives $X_i(\theta_i^L)\gs\beta_i$. Since $u$ is increasing and $u'(\beta)>0$, the maximum at the upper ability boundary is $\beta$ for sufficiently small $c$, and by continuity it remains $\beta$ on a nonempty upper ability interval. This proves the segmentation, allowing degenerate intervals and assigning cutoff values by the same tie convention.

For the comparison with complete information, reduce $c$ further if necessary so that its pointwise optimum is $\beta$ for every $\theta\gs\Phi_i^{-1}(0)/2$. Such a uniform bound exists: at $c=0$ the complete-information benefit is strictly increasing in exposure, and its secant slope toward $\beta$ has a strictly positive infimum, using compactness and $u'(\beta)>0$. Hence $\theta_i^h<\Phi_i^{-1}(0)$, whereas $\theta_i^L>\Phi_i^{-1}(0)$. The ordered cutoffs then imply all three strict comparisons in the proposition. 

Finally, choose any $\delta_i\in(\Phi_i^{-1}(0),\hi\theta)$. For all sufficiently small $c$, every $\theta\gs\delta_i$ is active, and the envelope implies $q_i(\theta)\gs\Phi_i^{-1}(0)u(\beta_i)$. Thus the quality delivered even to the lowest-valued category-$i$ viewer is bounded below by $\alpha_i\Phi_i^{-1}(0)u(\beta_i)(1-F_i(\delta_i))>0$, independently of sufficiently small $c$. Implement cross-category exposure uniformly as in \autoref{lem:sep}. Content attention costs are at most $c\alpha$ for every viewer, so the condition \eqref{eq:pos-inc} holds for all viewer types when $c$ is sufficiently small. Advertising is determined by the viewer envelope. Own-category content and the advertising load are independent of viewer value, so viewer IC holds. This completes verification of the relaxed upper bound.
\end{proof}

\begin{proof}[Proof of \autoref{cor:inc}]
By the retained-IC reduction in the proof of \autoref{prop:incomplete}, it suffices to maximize $\beta_i\mathbb E_i[\Psi_i(v)]\Phi_i(\theta)u(X)-cX$ over $X\in\{0\}\cup[\beta_i,\beta]$. Since $u$ is strictly concave with $u(0)=0$, $u(X)/X$ strictly decreases. The first positive exposure is therefore $\beta_i$, and indifference against zero gives $\Phi_i(\theta_i^L)=c/(u(\beta_i)\mathbb E_i[\Psi_i(v)])$. The largest-maximizer convention assigns $X_i(\theta_i^L)=\beta_i$.

On the extended-reach segment, strict concavity gives the interior condition
$$\beta_i\mathbb E_i[\Psi_i(v)]\Phi_i(\theta)u'(X_i(\theta))=c.$$
 Evaluating this condition at $X=\beta_i$ and $X=\beta$ gives $\theta_i^M$ and $\theta_i^H$, respectively, and inversion gives the displayed allocation. Strict concavity implies $\theta_i^L<\theta_i^M<\theta_i^H$ whenever all cutoffs are interior. The inactive and full-exposure branches are as in \autoref{prop:incomplete}.
\end{proof}

\begin{proof}[Proof of \autoref{prop:impc}]
Facing function $\chi_i(q)$, a $(\theta,i)$-creator maximizes:
$$u(\chi_i(q))-\frac{q}{\theta}.$$
Where the schedule is differentiable, the implicit function theorem gives $\chi_i'(q)=\frac{1}{X_i^{-1}(\chi_i(q))u'(\chi_i(q))}$, and the payoff derivative is $\frac{1}{X_i^{-1}(\chi_i(q))}-\frac{1}{\theta}$. To cover pooling and jumps as well, use the inverse-area identity:
$$\int_{(0,X_i(\theta)]}X_i^{-1}(X)\dd u(X)=\theta u(X_i(\theta))-\int_0^\theta u(X_i(\tilde{\theta}))\dd\tilde{\theta}=q_i(\theta).$$
For every $x\in[0,X_i(\hi{\theta})]$, the generalized inverse is no greater than $\theta$ below $X_i(\theta)$ and no less than $\theta$ above it. Therefore,
$$\int_{(0,x]}X_i^{-1}(X)\dd u(X)-q_i(\theta)\gs\theta\lt(u(x)-u(X_i(\theta))\rt).$$
Rearranging shows that $u(x)-\frac{1}{\theta}\int_{(0,x]}X_i^{-1}(X)\dd u(X)\ls u(X_i(\theta))-\frac{q_i(\theta)}{\theta}$, proving global optimality of the prescribed quality for every $\theta>0$. This argument includes the entry cliff and all pooling intervals. Quality above $\hi{q}_i$ produces no additional attention and costs more effort, so it cannot improve the creator's payoff.
\end{proof}

\begin{proof}[Proof of \autoref{prop:impv}]
Facing function $K_j(A)$, a $(v,j)$-viewer maximizes:
$$vK_j(A)-cA.$$
By construction of the direct mechanism, the $(v,j)$-viewer is willing to participate and does not mimic any $(v',j)$-type. Then:
\begin{eqnarray}
vK_j(A_j(v))-cA_j(v)\gs vQ_j(v)-cA_j(v)&\gs&\sum_k\rho_k(vQ_j(\hat{v}_k)-cA_j(\hat{v}_k))\nn\\
&=&v\sum_k\rho_kQ_j(\hat{v}_k)-c\sum_k\rho_kA_j(\hat{v}_k)\nn
\end{eqnarray}
for any $\{\rho_k\}_k\subset\mathbb{R}_+$ with $\sum_k\rho_k\ls1$ and $\{\hat{v}_k\}_k\subset[\lo{v},\hi{v}]$. By definition, this means:
$$vK_j(A_j(v))-cA_j(v)\gs vK_j(A)-cA$$
for all $A\in[0,A_j(\hi{v})]$, and $Q_j(v)=K_j(A_j(v))$. Therefore, total attention $A_j(v)$ is implemented for viewer $(v,j)$ even though $A\notin\{A_j(v)\}_v$ is allowed.

It remains to construct the physical layers of the feed. In the allocation characterized by \autoref{prop:incomplete} for sufficiently small $c$, all participating creators exhaust their own-category audiences. Therefore $Q_j(v)=Q_j(\lo{v})$, $W_j(v)=(v-\lo{v})Q_j(\lo{v})$, and $A_j(v)=\lo{v}Q_j(\lo{v})/c$ are constant across viewer types. Under the uniform foreign recommendations in \autoref{lem:sep}, every viewer receives the same content bundle and ad mass. Advertising feasibility ensures that this bundle can be arranged in an initial layer; proportional mixing generates the line from zero to its prescribed attention-quality point. All higher viewer types choose to consume the entire layer, while the lowest type is indifferent. Beyond this layer only ads or foreign content are offered.

More generally, the same construction applies to a nested cutoff allocation with partial own-category exposure if each creator newly entering the feed satisfies $G_j^{-1}(1-X_j(\theta)/\beta_j)q_j(\theta)\gs c$ for $0<X_j(\theta)<\beta_j$, and the uniform foreign implementation satisfies the advertising nonnegativity condition. Indeed, the viewer envelope gives $c\dd A_j(v)=v\dd Q_j(v)$, so for $v_2>v_1$,
\begin{eqnarray}
&&c\lt[A_j(v_2)-A_j(v_1)-\alpha_j\int_{\vartheta_j(v_2)\setminus\vartheta_j(v_1)}\dd F_j(\theta)\rt]\nn\\
&=&\alpha_j\int_{\vartheta_j(v_2)\setminus\vartheta_j(v_1)}\lt[G_j^{-1}\lt(1-\frac{X_j(\theta)}{\beta_j}\rt)q_j(\theta)-c\rt]\dd F_j(\theta)\gs0.\nn
\end{eqnarray}
Thus every attention increment accommodates the newly assigned own-category content. Uniform foreign content volume is constant across viewer types, and ads are
$$a_j(v)=A_j(v)-\alpha_j\int_{\vartheta_j(v)}\dd F_j(\theta)-\sum_{i\ne j}\alpha_i\int_\Theta\frac{R_i(\theta)}{\beta-\beta_i}\dd F_i(\theta).$$
The preceding inequality makes $a_j(v)$ nondecreasing, and its value at $\lo{v}$ is nonnegative by feasibility. Own-category content, foreign content, and ads can consequently all be nested. Form each layer from the nonnegative difference between successive bundles and mix proportionally across gaps. This establishes physical feasibility as well as IC in the model's cumulative recommendation representation; it does not assert that arbitrary cutoff menus can be delivered by a common deterministic sequence of indivisible items.
\end{proof}

\begin{proof}[Proof of \autoref{prop:welfare-com}]
Creator IR gives $e_i(\theta)\ls u(X_i(\theta))$. Increasing effort to this bound weakly increases all viewers' utilities. The resulting gross-surplus objective is \eqref{eq:obj}. Its relaxed maximizer is the allocation in \autoref{prop:complete}; by hypothesis, that allocation satisfies viewer IR. It therefore also maximizes the planner's objective. Binding viewer IR in the ad-funded implementation extracts the same gross surplus through ads.
\end{proof}

\begin{proof}[Proof of \autoref{prop:welfare}]
Write $q_i(\theta)=\theta e_i(\theta)$. We first drop viewer IC and IR and allow the planner to use ads, and verify these constraints below. Ads reduce this relaxed objective without benefiting creators, so they are set to zero.

Creator IC implies that $q_i(\theta)$ and $u(X_i(\theta))$ are nondecreasing. The envelope condition is $U_i'(\theta)=q_i(\theta)/\theta^2$, and $q_i(\theta)=\theta(u(X_i(\theta))-U_i(\theta))$ gives $(\theta U_i(\theta))'=u(X_i(\theta))$. Creator IR and nonnegative effort imply $0\ls U_i(\theta)\ls u(\beta)$, so:
\begin{equation}
U_i(\theta)=\frac1\theta\int_0^\theta u(X_i(\tilde{\theta}))\dd \tilde{\theta},\quad q_i(\theta)=\theta u(X_i(\theta))-\int_0^\theta u(X_i(\tilde{\theta}))\dd \tilde{\theta}=\int_{(0,\theta]}\tilde{\theta}\dd u(X_i(\tilde{\theta})).\label{eq:viewer-envelope-proof}
\end{equation}
Conversely, any bounded, nonnegative, nondecreasing $u(X_i(\cdot))$, together with these formulas, satisfies creator IC and IR and induces nonnegative effort. According to Fubini:
$$\int_\Theta q_i(\theta)\dd F_i(\theta)=\int_0^{\hi\theta}\Big(\theta f_i(\theta)-(1-F_i(\theta))\Big)u(X_i(\theta))\dd\theta,$$
and hence
\begin{equation}
\mathbb E_i[q_i(\theta)]=\mathbb E_i[\Phi_i(\theta)u(X_i(\theta))].\label{eq:viewer-creator-identity}
\end{equation}
Applying the same interchange of integrals to valuation-weighted quality
and differentiating its audience weights gives \eqref{eq:viewer-marginal}.

For the platform, \eqref{eq:obj-inc} is a virtual-surplus upper bound on normalized profit, since nonnegative bottom viewer utility reduces the amount that can be extracted. The following argument applies both to the viewer objective, with $\omega(v)=v$, and to this upper bound, with $\omega(v)=\Psi_i(v)$. Under the maintained assumptions these weights are positive and increasing, with a positive continuous limit at $\hi v$. For fixed exposure and nonnegative quality, first assign the highest-valued own-category viewers and use cross-category attention only after exhausting the own audience. Define the weighted audience
$$B_i(x)=\begin{cases}
\beta_i\int_{G_i^{-1}(1-x/\beta_i)}^{\hi v}\omega(v)\dd G_i(v),&0\ls x<\beta_i,\\
\beta_i\mathbb E_i[\omega(v)],&\beta_i\ls x\ls\beta.
\end{cases}$$
Category $i$'s contribution is proportional to $\int\lt(B_i(X_i(\theta))q_i(\theta)-cX_i(\theta)\rt)\dd F_i(\theta)$. Both $X_i(\cdot)$ and $B_i(X_i(\cdot))$ are nondecreasing.

Let $\theta_0$ be the unique root of $\Phi_i(\theta_0)=0$, and modify $X_i(\theta)=0$ for all $\theta<\theta_0$. This preserves creators' IC. Using \eqref{eq:viewer-envelope-proof}, the change to the objective is
$$\int_0^{\theta_0}u(X_i(t))\lt(\int_t^{\hi\theta}B_i(X_i(s))\dd F_i(s)-tB_i(X_i(t))f_i(t)\rt)\dd t+c\int_0^{\theta_0}X_i(t)\dd F_i(t)\gs0,$$
because the bracket is at least $-B_i(X_i(t))f_i(t)\Phi_i(t)\gs0$. Thus, it suffices to consider algorithms with no exposure at negative virtual types.

For $0\ls x<\beta_i$, write $\tau_i(x)=G_i^{-1}(1-x/\beta_i)$. Then $x\frac{\dd}{\dd x}\frac{B_i(x)}{x}=\omega(\tau_i(x))-\frac{B_i(x)}{x}$ and
$$\frac{\dd}{\dd x}\lt(u(x)\frac{B_i(x)}{x}\rt)=u'(x)\frac{B_i(x)}{x}+\frac{u(x)}x\lt(\omega(\tau_i(x))-\frac{B_i(x)}{x}\rt),$$
which converges to $u'(0)\omega(\hi v)>0$ as $x\downarrow0$. Hence there exists $\hat{\beta}\in(0,\beta_i)$, independent of $c$, such that $u(x)\frac{B_i(x)}{x}$ is increasing and $\frac{u(x)}{x}$ is nonincreasing on $(0,\hat{\beta}]$.

Draw a common $\beta'\sim U[0,\hat{\beta}]$ and modify $X_i(\theta)=\hat{\beta}\mathbbm{1}\{X_i(\theta)\gs \beta'\}$ for types $\theta$ with $X_i(\theta)<\hat{\beta}$. Every realization is nondecreasing and preserves zero exposure below $\theta_0$. Thus, attention costs are unchanged in expectation. For a type $\theta$ with $X_i(\theta)\in(0,\hat{\beta})$, the expected weighted-quality change reads:
\begin{multline*}
\lt(\frac{B_i(\hat{\beta})}{\hat{\beta}}u(\hat{\beta})-\frac{B_i(X_i(\theta))}{X_i(\theta)}u(X_i(\theta))\rt)\lt(\theta X_i(\theta)-\int_0^\theta X_i(\tilde{\theta})\dd \tilde{\theta}\rt)\\
+\frac{B_i(X_i(\theta))}{X_i(\theta)}\int_0^\theta\lt(X_i(\theta)u(X_i(\tilde{\theta}))-u(X_i(\theta))X_i(\tilde{\theta})\rt)\dd \tilde{\theta}\gs0.
\end{multline*}
The first term is nonnegative by local increase of $u(x)\frac{B_i(x)}{x}$ and monotonicity of $X_i(\cdot)$. The second is nonnegative by local decrease of $u(x)/x$. At types with $X_i(\theta)\gs\hat{\beta}$, current $u(X_i(\theta))$ is unchanged and $u(X_i(\theta'))$ for $\theta'<\theta$ weakly decreases in expectation, so expected quality weakly increases. Some deterministic realization weakly improves welfare and has $X_i\in\{0\}\cup[\hat{\beta},\beta]$. For every active $\theta$, we have $q_i(\theta)=\int_{(0,\theta]}\tilde{\theta}\dd u(X_i(\tilde{\theta}))\gs\theta_0u(X_i(\theta))\gs\theta_0u(\hat{\beta})>0$.

Finally, replace every positive $X_i(\theta)<\beta_i$ by $\beta_i$. This preserves creator IC. Write $\hat{X}_i$ for the exposure of the modified algorithm, $B_i^*=\beta_i\mathbb{E}_i[\omega(v)]$, and $\lo{\omega}=\min_v\omega(v)>0$. Because exposure changes only at nonnegative virtual abilities, \eqref{eq:viewer-creator-identity} implies $\mathbb E_i[\hat{q}_i-q_i]\gs0$. Every active creator now has its full own audience. The objective's change reads:
$$B_i^*\mathbb{E}_i[\hat{q}_i-q_i]+\mathbb{E}_i[(B_i^*-B_i(X_i))q_i]-c\mathbb{E}_i[\hat{X}_i-X_i]\gs[\lo{\omega}\theta_0u(\hat{\beta})-c]\mathbb E_i[\hat{X}_i-X_i],$$
and the inequality follows from $B_i^*-B_i(X_i)\gs \lo{\omega}(\hat{X}_i-X_i)$. For $c<\lo{\omega}\theta_0u(\hat{\beta})$, the improvement is strictly positive. Taking the minimum of these positive bounds over the finitely many categories and the two objectives, we establish the global reduction to full own-category audiences for sufficiently small $c$.

In \autoref{cor:inc}, the additional condition permits the rounding argument with $\hat{\beta}=\beta_i$, eliminating partial own-category exposure without a small-cost restriction on that reduction. For profit maximization this follows from the assumed increase of $u(x)\tau_i(x)$ on $(0,\beta_i]$. For viewer optimization, use $\frac{B_i(x)}{x}=\tau_i(x)+\mathbb{E}_i[v-\tau_i(x)\mid v\gs\tau_i(x)]$. Under an increasing hazard rate, mean residual life is nonincreasing in its valuation threshold; since $\tau_i(x)$ decreases in $x$, the second term is nondecreasing in $x$. Its product with positive increasing $u(x)$ is nondecreasing, so $u(x)\frac{B_i(x)}{x}$ increases on the full interval. Viewer IR and interior upper cutoffs still require the separate verification below under the maintained small-$c$ assumption.

Note that
\begin{equation}
X_i(\theta)\in\arg\max_{x\in\{0\}\cup[\beta_i,\beta]}\{\beta_i\mathbb E_i[v]\Phi_i(\theta)u(x)-cx\},\label{eq:pointwise}
\end{equation}
with the largest maximizer chosen at a tie. Equation \eqref{eq:viewer-creator-identity} makes \eqref{eq:pointwise} the pointwise maximization of the reduced objective. The objective in \eqref{eq:pointwise} has strictly increasing differences between distinct exposures, and therefore its largest maximizer is nondecreasing. Define $\theta_i^{\mathcal{L}}\equiv\inf\{\theta\mid X_i(\theta)>0\}$, $\theta_i^{\mathcal{M}}\equiv\inf\{\theta\mid X_i(\theta)>\beta_i\}$, and $\theta_i^{\mathcal{H}}\equiv\inf\{\theta\mid X_i(\theta)=\beta\}$. This gives the four (possibly degenerate) intervals in \autoref{prop:welfare}. Nonpositive virtual types receive no attention. For $c<\beta_i\mathbb E_i[v]\Phi_i(\hi\theta)\min_{[\beta_i,\beta]}u'$, we guarantee $\theta_i^{\mathcal{H}}<\hi\theta$. Creator IC follows from monotonicity and \eqref{eq:viewer-envelope-proof}. Specifically,
\begin{equation}
\begin{aligned}
x_{ii}(\theta,v)&=\mathbbm{1}\{X_i(\theta)>0\},& x_{ij}(\theta,v)&=\frac{(X_i(\theta)-\beta_i)^+}{\beta-\beta_i}&&(j\ne i),\\
e_i(\theta)&=\frac1\theta\int_{(0,\theta]}\tilde{\theta}\dd u(X_i(\tilde{\theta})),& a_j(v)&=0.
\end{aligned}\label{eq:implementation}
\end{equation}
Viewer IR requires
\begin{equation}
\lo v\,\alpha_j\int_\Theta\theta e_j(\theta)\dd F_j(\theta)\gs c\sum_i\alpha_i\int_\Theta x_{ij}(\theta,v)\dd F_i(\theta)\quad\mbox{for all }j, \nn
\end{equation}
Every viewer report yields the same content bundle within a category, so viewer IC holds. When the displayed condition holds, the mechanism is viewer-optimal even if advertising is allowed. To verify it for small $c$, let $\theta_0$ denote the unique zero of $\Phi_j$ in category $j$. For fixed $\theta>\theta_0$, the largest maximizer in \eqref{eq:pointwise} converges to $\beta$ as $c\downarrow0$, and it is zero below $\theta_0$. Dominated convergence in \eqref{eq:viewer-envelope-proof} gives $q_j(\theta)\lmt\theta_0u(\beta)\mathbbm{1}\{\theta>\theta_0\}$ almost everywhere. Thus the left-hand side of viewer IR converges to $\lo v\,\alpha_j\theta_0u(\beta)(1-F_j(\theta_0))>0$, whereas its right-hand side is bounded by $c\sum_i\alpha_i\lmt0$. The same argument applies to the ad-funded objective with weight $\mathbb E_j[\Psi_j(v)]=\lo v$; its pooled content allocation admits a nonnegative constant advertising load for sufficiently small $c$.
\end{proof}



\begin{proof}[Proof of \autoref{prop:welfare-comparison}]
The proof of \autoref{prop:welfare} reduces both objectives to full own-category audiences for sufficiently small $c$. Their respective pointwise objectives are \eqref{eq:pointwise}, with valuation weight $\mathbb E_i[v]$ for viewer optimization and $\mathbb E_i[\Psi_i(v)]=\lo v$ for profit maximization. Strict concavity gives $u(\beta_i)>\beta_i u'(\beta_i)$, so exposure jumps to $\beta_i$ at entry. The three cutoffs solve $\Phi_i(\theta_i^{\mathcal L})=c/(\mathbb E_i[v]u(\beta_i))$, $\Phi_i(\theta_i^{\mathcal M})=c/(\beta_i\mathbb E_i[v]u'(\beta_i))$, and $\Phi_i(\theta_i^{\mathcal H})=c/(\beta_i\mathbb E_i[v]u'(\beta))$, with the analogous equations replacing $\mathbb E_i[v]$ by $\lo v$ for profit maximization. All cutoffs are interior for sufficiently small $c$. Since $\mathbb E_i[v]>\lo v$ and $\Phi_i$ is strictly increasing, each viewer-optimal cutoff is strictly lower.

For positive virtual abilities, the pointwise objective has increasing differences between exposure and the valuation weight. Its largest maximizer is therefore weakly higher in the viewer-optimal allocation. Both mechanisms assign zero exposure at nonpositive virtual abilities, so \eqref{eq:viewer-creator-identity} implies weakly greater aggregate quality under viewer optimization. The viewer IC and IR verification for both implementations is given in the proof of \autoref{prop:welfare}.
\end{proof}

\begin{proof}[Proof of \autoref{cor:welfare-sol}]
Under the conditions of \autoref{cor:inc}, the rounding argument in the proof of \autoref{prop:welfare} eliminates partial own-category exposure without using small $c$ in that step. Strict concavity in \eqref{eq:pointwise} gives \eqref{eq:implementation} and the stated cutoff equations. Under the maintained small-$c$ assumption, the same proof verifies interior upper cutoffs and viewer IR; viewer IC follows because bundles do not depend on the viewer report. The allocation therefore attains the relaxed upper bound.
\end{proof}

\begin{proof}[Proof of \autoref{prop:money}]
Retain creator IC throughout. The two IC inequalities for any $\theta<\theta'$ imply that both $q_i(\theta)=\theta e_i(\theta)$ and $Y_i(\theta)$ are nondecreasing. The envelope condition gives
$$q_i(\theta)=\theta u(Y_i(\theta))-\int_0^\theta u(Y_i(\tilde\theta))\dd\tilde\theta=\int_{(0,\theta]}\tilde\theta\ \dd u(Y_i(\tilde\theta)),$$
which yields the stated effort formula and nonnegative effort.

Holding the entire reward schedule fixed preserves creator IC, effort, and quality. For each creator, assign own-category exposure first to viewers with the highest virtual values, as in the rearrangement argument for \autoref{prop:incomplete}. Since $\Psi_i$ is positive, own-category attention is exhausted before foreign attention is used. Using the same $J_i(X)$, the terms depending on exposure at a fixed promised reward are
$$q_i(\theta)J_i(X)+\frac{c}{z}\pi(X)-cX,\qquad 0\ls X\ls\beta,\quad\pi(X)\ls Y_i(\theta).$$
For $X<\beta_i$, their derivative is $q_i(\theta)J_i'(X)+c(\pi'(X)/z-1)>0$. Thus $Y_i(\theta)<\pi(\beta_i)$ implies $X_i(\theta)=\pi^{-1}(Y_i(\theta))$ and $T_i(\theta)=0$, whereas $Y_i(\theta)\gs\pi(\beta_i)$ implies $X_i(\theta)\gs\beta_i$. In the latter case the choice reduces to
$$X_i(\theta)=\max\operatorname*{arg\,max}_{\substack{\beta_i\ls X\ls\beta\\\pi(X)\ls Y_i(\theta)}}\{\pi(X)-zX\},\qquad T_i(\theta)=Y_i(\theta)-\pi(X_i(\theta)).$$
As $Y_i(\theta)$ increases, the feasible interval expands. Its greatest maximizing exposure cannot fall, proving monotone exposure globally without relaxing creator IC. At $X_i=\beta_i$, $\pi'(\beta_i)>z$ and continuity exclude positive cash: if $Y_i>\pi(\beta_i)$, a small feasible exposure increase would improve the objective. To establish the attention cliff and verify the global allocation, first apply the truncation step in the proof of \autoref{prop:welfare} to the nondecreasing utility schedule $u(Y_i(\theta))$ and weights $\omega(v)=\Psi_i(v)$. Removing rewards below $\Phi_i^{-1}(0)>0$ preserves IC, weakly raises quality at higher types, and also saves any monetary costs paid to the removed types. On the remaining types with $0<X_i<\beta_i$, transfers are zero, so apply the common randomization argument with $u$ replaced by $u\circ\pi$ and $\hat\beta=\beta_i/2$. It preserves expected exposure costs and removes exposures in $(0,\beta_i/2)$; monetary costs are unchanged because only zero-transfer types are rounded.

Every remaining participating creator has quality at least $\Phi_i^{-1}(0)u(\pi(\beta_i/2))$. Raise each positive exposure below $\beta_i$ to $\beta_i$, leave other rewards and transfers unchanged, and determine effort from the creator envelope. This preserves monotone rewards and creator IC. As in the full own-category audience step of \autoref{prop:welfare}, the gain in normalized profit is at least
$$\lt(\Psi_i(\lo v)\Phi_i^{-1}(0)u(\pi(\beta_i/2))-c\rt)\mathbb E_i[\hat X_i-X_i].$$
The bound follows because the change in expected quality is nonnegative, while each added unit of own-category exposure has virtual value at least $\Psi_i(\lo v)$ and is attached to a creator with the stated minimum quality. Consequently, under $c<\min_i\{\Psi_i(\lo v)\Phi_i^{-1}(0)u(\pi(\beta_i/2))\}$, an optimum can be selected with $X_i(\theta)\in\{0\}\cup[\beta_i,\beta]$. This proves a positive attention cliff while retaining IC and nonnegative effort.

In this allocation, $J_i(X_i(\theta))q_i(\theta)=\beta_i\lo v\,q_i(\theta)$ for every type, since inactive types have zero quality. Fubini and the creator envelope therefore reduce category $i$'s objective exactly to
$$\int_\Theta\lt(\beta_i\lo v\Phi_i(\theta)u(Y_i(\theta))-cX_i(\theta)-\tfrac{c}{z}T_i(\theta)\rt)\dd F_i(\theta).$$
By \autoref{lem:resource}, choose the greatest maximizer of
$$\max_{\hat X_i\gs\beta_i}\ \beta_i\lo v\Phi_i(\theta)u(\hat\pi(\hat X_i))-c\hat X_i$$
and compare its value with zero, corresponding to no participation. Thus the optimization also determines entry rather than holding a promised reward at entry fixed. The increasing hazard rate makes $\Phi_i$ increasing. This objective has increasing differences in $(\theta,\hat X_i)$ because $u\circ\hat\pi$ is increasing; its greatest maximizing resource, and therefore reward, is nondecreasing. To verify the lower boundary, reoptimize the entire tail starting when exposure first reaches $\beta_i$, including any pooling at $\beta_i$. Every lower type has $T_i=0$ and $Y_i=\pi(X_i)<\pi(\beta_i)$, whereas each tail choice has $\hat X_i\gs\beta_i$ and $Y_i\gs\pi(\beta_i)$. 

At an interior choice with $T_i>0$, the reward condition is $\beta_i\lo v\Phi_i(\theta)u'(Y_i)=c/z$. When $T_i=0$ and $\beta_i<X_i<\beta$, the relevant condition is instead $\beta_i\lo v\Phi_i(\theta)u'(\pi(X_i))\pi'(X_i)=c$; boundary choices satisfy the corresponding one-sided inequalities. The resource rule in \autoref{lem:resource} gives intervals of exposure provision, cash provision, or pooling, with possible upward jumps when the frontier is non-concave. No continuous passage through all contact points is required.

If $\pi$ is concave, $\pi(0)=u(0)=0$ and concavity imply that $u(\hat\pi(\hat X))/\hat X$ decreases on $[\beta_i,\infty)$. Entry therefore uses $\hat X_i=\beta_i$ and $T_i=0$, rather than jumping directly into a cash regime. The greatest constrained maximizer of $\pi(X)-zX$ is fixed. Before its reward is reached, pure exposure is optimal; afterwards, extra reward is supplied by money. If a positive-payment interval exists, its exposure is $X_i(\theta_i^\$)=\max\operatorname*{arg\,max}_{X\in[\beta_i,\beta]}\{\pi(X)-zX\}$. 
The interval is nonempty precisely when
$$z\beta_i\lo v\Phi_i(\hi\theta)u'\lt(\pi\lt(\max\operatorname*{arg\,max}_{X\in[\beta_i,\beta]}\{\pi(X)-zX\}\rt)\rt)>c.$$
If $\pi'(\beta)>z$, the capacity pooling regime has $X_i=\beta$, $T_i=0$, for abilities satisfying
$$\frac{c}{u'(\pi(\beta))\pi'(\beta)}\ls\beta_i\lo v\Phi_i(\theta)\ls\frac{c}{z\,u'(\pi(\beta))},$$
provided the support reaches this interval.

Finally, choose any $\delta_i\in(\Phi_i^{-1}(0),\hi\theta)$. For sufficiently small $c$, the feasible trial $(X_i,T_i)=(\beta_i,0)$ has strictly positive reduced objective at every $\theta\gs\delta_i$, so all these types are active. Because rewards are zero below $\Phi_i^{-1}(0)$ and every participating type receives at least $\pi(\beta_i)$, the envelope gives $q_i(\theta)\gs\Phi_i^{-1}(0)u(\pi(\beta_i))$ for these types. Hence $Q_i(\lo v)$ is at least $\alpha_i\Phi_i^{-1}(0)u(\pi(\beta_i))(1-F_i(\delta_i))>0$, independently of sufficiently small $c$. Implement foreign recommendations uniformly as in \autoref{lem:sep}. Every viewer's content cost is at most $c\alpha$, so $\lo v Q_j(\lo v)\gs c\alpha$ for all $j$ when $c$ is sufficiently small. The 
constraint \eqref{eq:pos-inc} then holds. Own-category quality, foreign content and the ad mass determined by the viewer envelope are independent of the viewer's report; viewer IC and IR follow. The allocation attains the relaxed upper bound and therefore solves the original problem.
\end{proof}

\begin{proof}[Proof of \autoref{lem:resource}]
Holding the promised reward $Y_i(\theta)\gs\pi(\beta_i)$ fixed preserves creator IC, effort, and quality. Once own-category attention is exhausted, the platform minimizes $X_i+T_i/z$ subject to $\beta_i\ls X_i\ls\beta$, $T_i\gs0$, and $\pi(X_i)+T_i=Y_i$. The definition of $\hat\pi$ gives the minimum $\hat\pi^{-1}(Y_i)$: the resource budget binds because any unused resource can increase money, and $\hat\pi$ is continuous and strictly increasing by \eqref{hat-c2}. Hence $Y_i=\hat\pi(\hat X_i)$ and $T_i=z(\hat X_i-X_i)=Y_i-\pi(X_i)$.

Solving this equality for feasible exposure and resolving ties toward greater exposure gives
$$X_i(\theta)=\max\lt\{X\in[\beta_i,\min\{\beta,\hat\pi^{-1}(Y_i(\theta))\}]\ \Big|\ \pi(X)+z\lt(\hat\pi^{-1}(Y_i(\theta))-X\rt)=Y_i(\theta)\rt\}.$$
\end{proof}

\begin{proof}[Proof of \autoref{prop:z}]
On extended reach, the relevant objective is $\beta_i\lo v\Phi_i(\theta)u(\pi(X_i)+T_i)-cX_i-cT_i/z$, whose reward coefficient is independent of $z$. If $\Phi_i(\theta)\ls0$, it is maximized at $(X_i,T_i)=(\beta_i,0)$, independently of $z$. Otherwise its reward coefficient is positive. Let $z_2>z_1$ and choose optimal pairs $(X_{ik},T_{ik})$, $k=1,2$, selecting the greatest exposure among all optimal pairs. Adding their two optimality inequalities gives
$$c\lt(\tfrac{1}{z_1}-\tfrac{1}{z_2}\rt)(T_{i2}-T_{i1})\gs0,$$
so $T_{i2}\gs T_{i1}$. At a fixed $T_i$, exposure maximizes $\beta_i\lo v\Phi_i(\theta)u(\pi(X_i)+T_i)-cX_i$ on $[\beta_i,\beta]$. Strict concavity of $u$ and increase of $\pi$ imply decreasing differences in $(X_i,T_i)$; where second derivatives exist, the cross partial is $\beta_i\lo v\Phi_i(\theta)u''(\pi(X_i)+T_i)\pi'(X_i)\ls0$. Hence its greatest maximizing exposure is nonincreasing in $T_i$. The selected global maximizer must also be the greatest conditional maximizer at its transfer. Therefore $X_{i2}\ls X_{i1}$; when transfers are equal, the same tie rule selects the same exposure. The argument applies within the maintained parameter range and when the creator remains on extended reach at both values.
\end{proof}

\bibliography{biblio}

@techreport{NgTaylor2025,
  author = {Ng, Robin and Taylor, Greg},
  title = {Moderating content-hosting platforms},
  year = {2025},
  number = {698},
  institution = {CRC TR 224, University of Bonn and University of Mannheim},
  type = {Discussion Paper},
  note = {Revised June 2026}
}

@techreport{BiscegliaBonattiScottMorton2025,
  author = {Bisceglia, Michele and Bonatti, Alessandro and {Scott Morton}, Fiona},
  title = {Regulating Privacy Policies on Digital Platforms},
  year = {2025},
  number = {20842},
  institution = {Centre for Economic Policy Research},
  type = {CEPR Discussion Paper},
  note = {Revised July 2026}
}

@article{AcemogluOzdaglarSiderius2024,
  author = {Acemoglu, Daron and Ozdaglar, Asuman and Siderius, James},
  title = {A model of online misinformation},
  journal = {Review of Economic Studies},
  year = {2024},
  volume = {91},
  number = {6},
  pages = {3117--3150},
  publisher = {Oxford University Press UK}
}

@article{ChoiJeon2023,
  author = {Choi, Jay Pil and Jeon, Doh-Shin},
  title = {Platform design biases in ad-funded two-sided markets},
  journal = {The RAND Journal of Economics},
  year = {2023},
  volume = {54},
  number = {2},
  pages = {240--267},
  publisher = {Wiley Online Library}
}

@book{borgers2015introduction,
  title={An introduction to the theory of mechanism design},
  author={B{\"o}rgers, Tilman},
  year={2015},
  publisher={Oxford university press}
}

@unpublished{BarIsaacDebMitchell2025,
  author = {Bar-Isaac, Heski and Deb, Rahul and Mitchell, Matthew},
  title = {Selling certification, content moderation, and attention},
  year = {2025},
  note = {arXiv:2506.12604}
}

@article{deCorniereSarvary2023,
  author = {de Corniere, Alexandre and Sarvary, Miklos},
  title = {Social media and news: Content bundling and news quality},
  journal = {Management Science},
  year = {2023},
  volume = {69},
  number = {1},
  pages = {162--178},
  publisher = {INFORMS}
}

@incollection{AndersonJullien2015,
  author = {Anderson, Simon P and Jullien, Bruno},
  title = {The advertising-financed business model in two-sided media markets},
  booktitle = {Handbook of media economics},
  year = {2015},
  volume = {1},
  pages = {41--90},
  publisher = {Elsevier}
}

@unpublished{IchihashiJeonKim2026,
  author = {Ichihashi, Shota and Jeon, Doh-Shin and Kim, Byung-Cheol},
  title = {Mechanism Design for Ad-Supported Platforms},
  year = {2026},
  note = {Working paper; arXiv:2609.12870}
}

@article{BergemannBonatti2024,
  author = {Bergemann, Dirk and Bonatti, Alessandro},
  title = {Data, competition, and digital platforms},
  journal = {American Economic Review},
  year = {2024},
  volume = {114},
  number = {8},
  pages = {2553--2595},
  publisher = {American Economic Association 2014 Broadway, Suite 305, Nashville, TN 37203}
}

@incollection{JullienPavanRysman2021,
  author = {Jullien, Bruno and Pavan, Alessandro and Rysman, Marc},
  title = {Two-sided markets, pricing, and network effects},
  booktitle = {Handbook of Industrial Organization},
  year = {2021},
  volume = {4},
  pages = {485--592},
  publisher = {Elsevier}
}

@article{ToubiaStephen2013,
  author = {Toubia, Olivier and Stephen, Andrew T},
  title = {Intrinsic vs. image-related utility in social media: Why do people contribute content to twitter?},
  journal = {Marketing Science},
  year = {2013},
  volume = {32},
  number = {3},
  pages = {368--392},
  publisher = {INFORMS}
}

@article{AndersonCoate2005,
  author = {Anderson, Simon P and Coate, Stephen},
  title = {Market provision of broadcasting: A welfare analysis},
  journal = {The Review of Economic Studies},
  year = {2005},
  volume = {72},
  number = {4},
  pages = {947--972},
  publisher = {Wiley-Blackwell}
}

@article{QianJain2024,
  author = {Qian, Kun and Jain, Sanjay},
  title = {Digital content creation: An analysis of the impact of recommendation systems},
  journal = {Management Science},
  year = {2024},
  volume = {70},
  number = {12},
  pages = {8668--8684},
  publisher = {INFORMS},
  doi = {10.1287/mnsc.2022.03655}
}

@article{AridorJimenezDuranLevySong2024,
  author = {Aridor, Guy and Jim{\'e}nez-Dur{\'a}n, Rafael and Levy, Ro’ee and Song, Lena},
  title = {The economics of social media},
  journal = {Journal of Economic Literature},
  year = {2024},
  volume = {62},
  number = {4},
  pages = {1422--1474},
  publisher = {American Economic Association 2014 Broadway, Suite 305, Nashville, TN 37203-2425}
}

@article{hebert2026engagement,
  title={Engagement maximization},
  author={H{\'e}bert, Benjamin and Zhong, Weijie},
  journal={Theoretical Economics},
  volume={21},
  number={3},
  pages={849--887},
  year={2026},
  publisher={Wiley Online Library}
}

@inproceedings{IchihashiSmolin2025,
  author = {Ichihashi, Shota and Smolin, Alex},
  title = {Buyer-Optimal Algorithmic Recommendations},
  booktitle = {Proceedings of the 26th ACM Conference on Economics and Computation},
  year = {2025},
  pages = {667},
  doi = {10.1145/3736252.3742609}
}

@article{Myerson1986,
  author = {Myerson, Roger B},
  title = {Multistage games with communication},
  journal = {Econometrica},
  year = {1986},
  volume = {54},
  number = {2},
  pages = {323--358},
  publisher = {JSTOR}
}

@article{Norman2004,
  author = {Norman, Peter},
  title = {Efficient mechanisms for public goods with use exclusions},
  journal = {The Review of Economic Studies},
  year = {2004},
  volume = {71},
  number = {4},
  pages = {1163--1188},
  publisher = {Wiley-Blackwell}
}

@article{CasnerTeh2025,
  author = {Casner, Ben and Teh, Tat-How},
  title = {Content-Hosting Platforms: Discovery, Membership, Or Both?},
  journal = {The RAND Journal of Economics},
  year = {2025},
  volume = {56},
  number = {4},
  pages = {668--688},
  publisher = {Wiley Online Library},
  doi = {10.1111/1756-2171.70029}
}

@article{DamianoLi2007,
  author = {Damiano, Ettore and Li, Hao},
  title = {Price discrimination and efficient matching},
  journal = {Economic Theory},
  year = {2007},
  volume = {30},
  number = {2},
  pages = {243--263},
  publisher = {Springer}
}

@article{GomesPavan2024,
  author = {Gomes, Renato and Pavan, Alessandro},
  title = {Price customization and targeting in matching markets},
  journal = {The RAND Journal of Economics},
  year = {2024},
  volume = {55},
  number = {2},
  pages = {230--265},
  publisher = {Wiley Online Library}
}

@article{GomesPavan2016,
  author = {Gomes, Renato and Pavan, Alessandro},
  title = {Many-to-many matching and price discrimination},
  journal = {Theoretical Economics},
  year = {2016},
  volume = {11},
  number = {3},
  pages = {1005--1052},
  publisher = {Wiley Online Library}
}

@article{JohnsonRhodesWildenbeest2023,
  author = {Johnson, Justin P and Rhodes, Andrew and Wildenbeest, Matthijs},
  title = {Platform design when sellers use pricing algorithms},
  journal = {Econometrica},
  year = {2023},
  volume = {91},
  number = {5},
  pages = {1841--1879},
  publisher = {Wiley Online Library}
}

@article{Chen2026,
  author = {Chen, Daniel},
  title = {A Model of the Attention Economy: Data, Platforms, and Discovery},
  journal = {American Economic Review},
  year = {2026},
  volume = {116},
  number = {10},
  pages = {3802--3836},
  doi = {10.1257/aer.20240651}
}

@article{AmbrusCalvanoReisinger2016,
  author = {Ambrus, Attila and Calvano, Emilio and Reisinger, Markus},
  title = {Either or both competition: A “two-sided” theory of advertising with overlapping viewerships},
  journal = {American Economic Journal: Microeconomics},
  year = {2016},
  volume = {8},
  number = {3},
  pages = {189--222},
  publisher = {American Economic Association 2014 Broadway, Suite 305, Nashville, TN 37203-2425}
}

@article{BergemannBonatti2011,
  author = {Bergemann, Dirk and Bonatti, Alessandro},
  title = {Targeting in advertising markets: implications for offline versus online media},
  journal = {The RAND Journal of Economics},
  year = {2011},
  volume = {42},
  number = {3},
  pages = {417--443},
  publisher = {Wiley Online Library}
}

@article{PratValletti2022,
  author = {Prat, Andrea and Valletti, Tommaso},
  title = {Attention oligopoly},
  journal = {American Economic Journal: Microeconomics},
  year = {2022},
  volume = {14},
  number = {3},
  pages = {530--557},
  publisher = {American Economic Association 2014 Broadway, Suite 305, Nashville, TN 37203-2425}
}

@article{OrenSmithWilson1982,
  author = {Oren, Shmuel S and Smith, Stephen A and Wilson, Robert B},
  title = {Nonlinear pricing in markets with interdependent demand},
  journal = {Marketing Science},
  year = {1982},
  volume = {1},
  number = {3},
  pages = {287--313},
  publisher = {INFORMS}
}

@article{Lockwood2000,
  author = {Lockwood, Ben},
  title = {Production externalities and two-way distortion in principal-multi-agent problems},
  journal = {Journal of Economic Theory},
  year = {2000},
  volume = {92},
  number = {1},
  pages = {142--166},
  publisher = {Elsevier}
}

@article{Hellwig2003,
  author = {Hellwig, Martin F},
  title = {Public-good provision with many participants},
  journal = {The Review of Economic Studies},
  year = {2003},
  volume = {70},
  number = {3},
  pages = {589--614},
  publisher = {Wiley-Blackwell}
}

@techreport{MeisnerPillath2024,
  author = {Meisner, Vincent and Pillath, Pascal},
  title = {Monetizing digital content with network effects: A mechanism-design approach},
  year = {2024},
  number = {0049},
  institution = {Berlin School of Economics},
  type = {Discussion Paper}
}

@unpublished{Srinivasan2024,
  author = {Srinivasan, Karthik},
  title = {Paying attention},
  year = {2024},
  note = {Working paper}
}

@article{anderson2023ad,
  title={Ad clutter, time use, and media diversity},
  author={Anderson, Simon P and Peitz, Martin},
  journal={American Economic Journal: Microeconomics},
  volume={15},
  number={2},
  pages={227--270},
  year={2023},
  publisher={American Economic Association, 2014 Broadway, Suite 305, Nashville, TN 37203~…}
}

@article{ValenzuelaStookey2026,
  author  = {Valenzuela-Stookey, Quitz{\'e}},
  title   = {Platform-Mediated Competition},
  journal = {American Economic Journal: Microeconomics},
  year    = {2026},
  note    = {Forthcoming},
  doi     = {10.1257/mic.20240475}
}

@article{rochet2006two,
  title={Two-sided markets: a progress report},
  author={Rochet, Jean-Charles and Tirole, Jean},
  journal={The RAND journal of economics},
  volume={37},
  number={3},
  pages={645--667},
  year={2006},
  publisher={Wiley Online Library}
}

@article{Caillaud2003,
  author  = {Caillaud, Bernard and Jullien, Bruno},
  title   = {Chicken \& Egg: Competition among Intermediation Service Providers},
  journal = {{RAND} Journal of Economics},
  year    = {2003},
  volume  = {34},
  number  = {2},
  pages   = {309--328},
  doi     = {10.2307/1593720}
}

@article{corrao2023nonlinear,
  title={Nonlinear pricing with underutilization: A theory of multi-part tariffs},
  author={Corrao, Roberto and Flynn, Joel P and Sastry, Karthik A},
  journal={American Economic Review},
  volume={113},
  number={3},
  pages={836--860},
  year={2023},
  publisher={American Economic Association, 2014 Broadway, Suite 305, Nashville, TN 37203}
}

@misc{Instagram2022,
  author = {{Instagram}},
  title = {Two New Ways to Control Your {Instagram} Feed},
  year = {2022},
  month = mar,
  howpublished = {\url{https://about.fb.com/news/2022/03/two-new-ways-to-control-your-instagram-feed/}},
  note = {March 23; announcement by Adam Mosseri}
}
\bibliographystyle{chicago}

\clearpage
\setcounter{section}{0}
\setcounter{subsection}{0}
\pagenumbering{arabic}

\renewcommand{\thesection}{OA.\arabic{section}}
\renewcommand{\thesubsection}{\thesection.\arabic{subsection}}
\renewcommand{\thepage}{OA-\arabic{page}}

\begin{center}{\LARGE\bfseries Online Appendix}\\[0.5em]\end{center}

This Online Appendix provides numerical illustrations showing that our main results do not hinge on several simplifying assumptions. Although relaxing these assumptions makes sharp analytical characterizations difficult, the underlying economic forces and qualitative predictions remain robust.


All examples assume complete information. The common primitives are $u(X)=\log(1+X)$, $\theta\sim U[0,1]$, $v\sim U[1,2]$, and $z=1$. Reported allocations are numerical approximations.

\section{Heterogeneous Out-of-category Viewer Values}
\label{oa-ex1}

\begin{example}[Discounted Consumption Values]\ \\
A category-$j$ viewer obtains gross utility $d_{ij}vq$ from consuming a
category-$i$ creator's content of quality $q$. Let
$$D_v\equiv(d_{ij})=\begin{pmatrix}
1&0.55&0.20\\
0.25&1&0.65\\
0.50&0.30&1
\end{pmatrix},\qquad
\begin{aligned}
(\alpha_1,\alpha_2,\alpha_3)&=(0.8,0.5,0.3),\\
(\beta_1,\beta_2,\beta_3)&=(1.2,0.8,0.5),
\end{aligned}
\qquad c=0.3.$$
For a given total attention $X_i$, select viewers in descending order of $d_{ij}v$. 
The inclusion thresholds are respectively $(0.2220,0.2744,0.2974)$ for the three categories, and exposure jumps at entry to $(2.0000,1.7818,2.3673)$. The modified condition \eqref{eq:pos} is slack in all three viewer categories.
\end{example}

\paragraph{Remarks.}
The attention cliff continues to appear. The pecking order of audience assignment now follows discounted valuation $d_{ij}v$, which need not place every own-category viewer ahead of every cross-category viewer. In this example, creators in all three categories receive extended reach immediately upon entry. Some recommended pairs nevertheless satisfy $d_{ij}v\theta u(X_i(\theta))<c$, and as is in the main model, negative direct surplus is justified by the quality improvement that caters to other viewers. Thus, attention subsidies persist even when cross-category consumption provides direct value.

\section{Heterogeneous Out-of-category Creator Utilities}
\label{oa-ex2}

\begin{example}[Discounted Attention Utilities]\ \\
A category-$i$ creator obtains utility $u(X_i^W(\theta))$ from attention, where
$$X_i^W(\theta)=\sum_j\delta_{ij}\beta_j\int_V x_{ij}(\theta,v)\dd G_j(v).$$
Let
$$D_c\equiv(\delta_{ij})=\begin{pmatrix}
1&0.55&0.20\\
0.25&1&0.65\\
0.50&0.30&1
\end{pmatrix},\qquad
\begin{aligned}
(\alpha_1,\alpha_2,\alpha_3)&=(0.8,0.5,0.3),\\
(\beta_1,\beta_2,\beta_3)&=(1.2,0.8,0.5),
\end{aligned}
\qquad c=0.15.$$
For a given weighted total attention $X_i^W$, select own-category viewers in descending order of $v$, followed by outside categories in descending order of
$\delta_{ij}$. 
The inclusion thresholds are respectively $(0.1268,0.1701,0.2466)$ for the three categories, and exposure jumps at entry to $(1.2,0.8,0.5)$. The modified condition \eqref{eq:pos} is slack in all three viewer categories.
\end{example}

\paragraph{Remarks.}
Own-category attention generates consumption surplus and has the greatest
incentive value, thus preserving its priority. Among out-of-category audiences, a higher $\delta_{ij}$ reduces the attention cost of delivering a given creator reward. In this example, creators in categories 1, 2, and 3 receive extended reach from viewers in categories 2, 3, and 1, beginning at abilities approximately $(0.3333,0.3462,0.6000)$, respectively. No other cross-category assignments are used. All such recommendations have direct surplus $-c$; their value comes from the quality improvement that additional exposure induces.

\section{Convex Attention Cost}
\label{oa-ex3}

\begin{example}[Competition for Attention]\ \\
Consider two symmetric categories with $\beta_i=1$. A viewer who spends total
attention $A$ on content and advertising incurs total cost
$$C(A)=\frac{1}{5}\lt(-A-2\log\lt(1-\tfrac A2\rt)\rt),\qquad 0\ls A<2.$$
This cost is strictly convex, with $C'(0)=0$ and $\lim_{A\lmt2}C'(A)=\infty$. Let $H_j(v)$ and $B_j(v)$ denote content attention and gross consumption benefit. Binding viewer participation gives $C(H_j(v)+a_j(v))=B_j(v)$, so advertising is $a_j(v)=C^{-1}(B_j(v))-H_j(v)$. Allocations maximize
$$\sum_j\beta_j\int_V\lt(C^{-1}(B_j(v))-H_j(v)\rt)\dd G_j(v)$$
subject to creator participation, recommendation bounds, and $H_j(v)<2$ and $B_j(v)\gs C(H_j(v))$ for all $v,j$. The last condition ensures nonnegative advertising.
With $\alpha_i=0.5$, $\theta=0.5$ starts to enter production and receives total attention of $1$. Type $\theta=0.72$ starts to receive extended reach.
\par\smallskip\noindent
\begin{minipage}{\linewidth}
Increasing creator mass in both categories gives:
\begin{center}
\begin{tabular}{lcc}
\toprule
& $\alpha_i=0.50$ & $\alpha_i=0.75$\\
\midrule
Fraction of creator types receiving exposure & $50\%$ & $37.5\%$\\
Exposure at $\theta\approx0.854$ & $1.40$ & $1.05$\\
Exposure at $\theta\approx0.979$ & $1.75$ & $1.35$\\
Average content attention per viewer & $0.306$ & $0.306$\\
\bottomrule
\end{tabular}
\end{center}
\end{minipage}
\end{example}

\paragraph{Remarks.}
More creators compete for limited total supply of attention, raising the inclusion threshold and reducing each creator's exposure. Because $C^{-1}$ is concave, viewer valuation generated by a creator's costly effort has diminishing returns to scale. Under linear attention costs (main model), a common scaling of creator masses merely scales profits and the \eqref{eq:pos} constraints, leaving optimal recommendation probabilities unchanged. Convex costs break this invariance and generate competition through rising marginal attention costs below the cap. The optimal allocation continues to use attention subsidy to encourage production.

\section{Large Attention Cost}

\label{oa-ex4}

\begin{example}[Binding Constraint \eqref{eq:pos}]\ \\
Consider two categories with constant marginal attention cost $c=0.3$ and contrasting sizes:
$$(\alpha_1,\alpha_2)=(2,0.4),\qquad (\beta_1,\beta_2)=(1,0.4).$$
Explicitly impose \eqref{eq:pos} for every $v\in[1,2]$. 
Each category-2 viewer receives zero ads as their \eqref{eq:pos} is binding. Each category-1 viewer $v$ has advertising surplus $ca_1(v)=0.70456v-0.47882>0$, corresponding to $a_1(v)=(0.70456v-0.47882)/c$ units of ads, as \eqref{eq:pos} is slack.
\end{example}

\paragraph{Remarks.}
Category-2 viewers use all surplus to cover the attention cost of exposure that incentivizes category-1 creators, while category-1 viewers continue to generate advertising revenue. Higher type viewers in category 2 derive higher utility and therefore receive more cross-category content, yet they are too few in number to sufficiently subsidize category-1 creators.

\end{document}